%% file: main.tex
\documentclass[conference]{IEEEtran}

\usepackage[T1]{fontenc}
\usepackage{silence}
\usepackage[dvipsnames]{xcolor}
\usepackage{amssymb}
\usepackage{amsthm}
\usepackage{mathtools}
\usepackage{mathpartir}
\usepackage{stmaryrd}
\usepackage{listings}
\usepackage{subcaption}
\usepackage{lipsum}
\usepackage{multirow}
\usepackage{enumitem}
\usepackage{tikz}
\usetikzlibrary{arrows,shapes,positioning,shadows,trees,tikzmark,chains,quotes,calc,decorations.pathreplacing,calligraphy}
\usepackage{array}
\usepackage[normalem]{ulem}
\usepackage{cleveref}
\crefname{section}{§}{§§}
\Crefname{section}{§}{§§}
\crefformat{section}{§#2#1#3}

\newtheorem*{intuition*}{Key Intuition}

\newtheorem{lemma}{Lemma}
\newtheorem{definition}{Definition}
\newtheorem{theorem}{Theorem}

\begin{document}

\include{macros}


\lstset{language=aara,basicstyle=\small}

\title{Automatic Amortized Resource Analysis with Regular Recursive Types}

\author{\IEEEauthorblockN{Jessie Grosen}
\IEEEauthorblockA{jgrosen@cs.cmu.edu\\
Carnegie Mellon University
}
\and
\IEEEauthorblockN{David M. Kahn}
\IEEEauthorblockA{davidkah@andrew.cmu.edu\\
Carnegie Mellon University
}
\and
\IEEEauthorblockN{Jan Hoffmann}
\IEEEauthorblockA{jhoffmann@cmu.edu\\
Carnegie Mellon University
}}

\IEEEoverridecommandlockouts
\IEEEpubid{\makebox[\columnwidth]{979-8-3503-3587-3/23/\$31.00~
\copyright2023 IEEE \hfill} \hspace{\columnsep}\makebox[\columnwidth]{ }}




\maketitle

\begin{abstract}

    The goal of automatic resource bound analysis is to statically
    infer symbolic bounds on the resource consumption of the
    evaluation of a program.
    A longstanding challenge for automatic resource analysis is
    the inference of bounds that are functions of complex custom data
    structures.
    This article builds on type-based automatic amortized resource
    analysis (AARA) to address this challenge.
    AARA is based on the potential method of amortized analysis and
    reduces bound inference to standard type inference with additional
    linear constraint solving, even when deriving non-linear bounds.
    A key component of AARA are resource functions that generate the
    space of possible bounds for values of a given type while enjoying
    necessary closure properties.

    Existing work on AARA defined such functions for many
    data structures such as lists of lists but the question of whether
    such functions exist for arbitrary data structures remained open.
    This work answers this questions positively by uniformly
    constructing resource polynomials for algebraic data
    structures defined by regular recursive types.
    These functions are a generalization of all previously proposed
    polynomial resource functions and can be seen as a general notion
    of polynomials for values of a given recursive type.
    A resource type system for FPC, a core language with recursive
    types, demonstrates how resource polynomials can be integrated
    with AARA while preserving all benefits of past techniques.
    The article also proposes the use of new techniques useful for
    stating the rules of this type system succinctly and proving it
    sound against a small-step cost semantics.
    First, multivariate potential annotations are stated in terms of
    free semimodules, substantially abstracting details of the
    presentation of annotations and the proofs of their properties.
    Second, a logical relation giving semantic meaning to resource
    types enables a proof of soundness by a single induction on
    typing derivations.
%
%
%
%
\end{abstract}

\section{Introduction}
\label{sec:introduction}

Programming language support for statically deriving resource (or
cost) bounds has been extensively studied.
Existing techniques encompass manual and automatic resource analyses
and are based on type
systems~\cite{dal2011linear,dal2013geometry,rajani2021unifying},
deriving and solving recurrence
relations~\cite{albert2007costa,kavvos2019recurrence,cutler2020denotational},
or other static analyses~\cite{gulwani2009speed,avanzini2015analysing,chatterjee2019non}.
They can derive (worst-case) upper bounds~\cite{wang2017timl,kincaid2017compositional} (best-case) lower
bounds~\cite{AlbertGM13,NgoDFH16}, and relational bounds on the difference of the cost of
two programs~\cite{radivcek2017monadic}, considering resources like time or memory.

Most automatic techniques focus on bounds that are functions of
integers or sizes of simple data structures like lists of integers.
One exception is automatic amortized resource analysis (AARA)~\cite{hofmann2003static,jost2010static,hoffmann2011multivariate},
which can automatically derive
bounds for complex data structures like lists of lists, taking
into account the individual lengths of inner lists.
As an example, consider the function \lstinline{sort_lefts_list},
which extracts only the left injections from its input list and sorts
the result.
Assume we are interested in the number of cons cells that are created
during the evaluation.
\begin{lstlisting}[xleftmargin=.4in]
let sort_lefts_list (l : (int + bool) list) =
  quicksort (filter_map find_left l)
\end{lstlisting}

RaML~\cite{hoffmann2017towards}, an implementation of 
AARA, is able to derive the exact worst-case bound 
of \(n^2 + n\) cons cell creations where \(n\) is only 
the number of left injections in the list. This small 
example highlights several key qualities of AARA: it is 
able to tightly analyze tricky recursion patterns, like 
those that appear in quicksort; it is compositional, 
easily handling interprocedural code; it produces exact, 
not asymptotic, polynomial bounds; and it can derive 
bounds on functions over tree-like data structures that 
take into account the shape of the data.

AARA for functional programs is based on a type system and type
derivations serve as proof certificates for the derived bounds.
Type inference is reduced to efficient linear programming and AARA
naturally derives bounds on the high-water mark resource use of
non-monotone resources like memory which can become available during
the evaluation.
The key innovation that enables inference of non-linear bounds with
linear programming is the use of 
a carefully selected set of resource functions that
serve as templates for the potential 
functions used in the physicist's method of amortized analysis.

Despite its benefits, state-of-the-art AARA still has some
limitations to its real-world applicability, including its lack of support for general, regular recursive types.\footnote{We use the term \emph{regular recursive types} to refer to types that may contain non-trivial nested recursion, but where all recursion occurs at base kind. 
} As an example, examine the function \lstinline{sort_lefts_tree}, similar to the function above but with lists swapped for \emph{rose trees}:

\begin{lstlisting}[xleftmargin=.0in]
let quicksort : int list -> int list = ...
type 'a tree = Tree of 'a * 'a tree list
let filter_map_tree : ('a -> 'b option) -> 'a tree -> 'b list = ...
let sort_lefts_tree (t : (int + bool) tree) =
  quicksort (filter_map_tree find_left t)
\end{lstlisting}

Rose trees can have arbitrary and variable branching factors, enabled by defining trees and lists of child trees in a nested fashion. Existing AARA systems cannot derive a bound for this function.
AARA's inability to derive bounds that are functions of general
algebraic data structures poses a real deficiency.
Extending polynomial AARA to handle regular recursive types has been an open problem since it was introduced by Hoffmann and Hofmann in 2010~\cite{hoffmann2010amortized}. The core challenge lies in finding a class of potential functions for these types that is expressive but constrained enough to be closed under the operations necessary for typing. 


We address this longstanding gap by introducing a notion of resource polynomials for regular recursive types that meets the requirements of AARA. We draw inspiration from past approaches, but ultimately adopt a more algebraic view that we believe better follows the structure of types. In particular, the indices that generate the base polynomials match the values they classify nearly exactly.
Our resource polynomials are a generalization of all previously proposed
polynomial resource functions of AARA~\cite{jost2010static,hoffmann2010amortized,hoffmann2011multivariate,hoffmann2017towards} and can be seen as a general notion of
polynomials for values of a given recursive type.
We give the two constructions, \emph{shifting} and \emph{sharing}, which witness resource polynomials' closure under discrete difference and multiplication, respectively; together, they enable AARA's inference of resource bounds using only linear programming. We describe these and other operations as linear maps on free semimodules in order to abstract away some of the tedious details in previous presentations. 
Finally, we build a type system for a version of FPC (a call-by-value language with recursive types~\cite{fiore1994axiomatization}) enriched with resource usage that makes use of these resource polynomials and prove it sound via a logical relations argument.

\section{Overview}
\label{sec:overview}

To start with, we review AARA~(\cref{sec:aaraintro}), detail its potential functions for
lists~(\cref{sec:listpotential}), and present the intuition behind our
extension to regular recursive types~(\cref{sec:nestedpotential}).

\subsection{A quick introduction to AARA}
\label{sec:aaraintro}

AARA is a type-based technique for automatically inferring worst-case
cost bounds for programs that manipulate data structures.
It uses a formalization of the physicist's method introduced by
Tarjan and Sleator~\cite{tarjan1985amortized} to assign potential functions to data
structures that can then be used for amortized analysis.
The potential available in a given context is then tracked across the
program to ensure that the available potential is sufficient to cover
the cost of the next transition and the potential of the resulting
state.

To automate the physicist's method, AARA defines a set of fixed
potential functions for each type.
These potential functions have to satisfy certain (closure) properties
that enable a smooth integration of potential tracking with the typing
rules.
This integration is the key to automation, because the potential
tracking can be expressed with linear constraints that can be
generated in tandem with type checking or inference. These constraints
can then be solved by an LP solver, resulting in a final type
annotated with a resource bound.



\subsubsection*{Example: \lstinline{filter_map}}
To demonstrate the basics of the AARA approach, we build up
the motivating example shown in the introduction.
As then, say we are interested in the number of cons cell creations as our cost model. To start, consider the standard list function \(\mathsf{filter\_map} : (\tau \to \mathsf{option}(\sigma)) \to \mathsf{list}(\tau) \to \mathsf{list}(\sigma)\), which is implemented as follows:
\begin{lstlisting}[xleftmargin=.4in]
let rec filter_map f l = match l with
  | [] -> []
  | x :: l' -> match f x with
            | Some y -> y :: filter_map f l'
            | None -> filter_map f l'
\end{lstlisting}
The evaluation of the expression \lstinline{filter_map f l} applies \lstinline{f} to each element of \lstinline{l} and collects the \lstinline{Some} results into the output list.
The cost of the evaluation depends on the cost of the higher-order argument \lstinline{f}.
First assume that the cost of \lstinline{f} is $0$.
Then the cost of \lstinline{filter_map f l} is, at worst, the length $|l|$ of the list \lstinline{l}.
This bound can be expressed by the following type.
%
\begin{align*}
  \mathsf{filter\_map} :\; &(\langle \mathsf{int}^0 + \mathsf{bool}^0, 0 \rangle \to \langle \mathsf{option}^0(\mathsf{int}), 0\rangle) \to \\
  &\langle \mathsf{list}^1(\mathsf{int}^0 + \mathsf{bool}^0), 0 \rangle \to 
  \langle \mathsf{list}^0(\mathsf{int}), 0 \rangle
\end{align*}
The type
$\langle int^0 + bool^0, 0 \rangle \to \langle \mathsf{option}^0(int), 0\rangle$
of the higher-order argument states that the function does not need
any input potential and does not assign any potential to its output.
The list type $\mathsf{list}^1(int^0 + bool^0)$ expresses that the list
argument carries one potential unit per element of the list,
reflecting the bound to be proved.
The output potential $\langle \mathsf{list}^0(int), 0 \rangle$ is zero
in this case but is, in general, important for the compositionality of the
analysis.
To see how the potential of the result can be used consider the
following typing:
%
\begin{align*}
  \mathsf{filter\_map} :\; &(\langle \mathsf{int}^1 + \mathsf{bool}^0, 0 \rangle \to \langle \mathsf{option}^1(\mathsf{int}), 0\rangle) \to \\
  &\langle \mathsf{list}^1(\mathsf{int}^1 + \mathsf{bool}^0), 0 \rangle \to 
  \langle \mathsf{list}^1(\mathsf{int}), 0 \rangle
\end{align*}
Here the resulting list carries $1$ potential unit per element.
To cover this additional potential, the input list now has type
$\mathsf{list}^1(int^1 + bool^0)$, which expresses $1$ potential unit
per element and one additional potential unit for each element of the
form $\inlex{n}$.
The type of the higher-order argument expresses that $1$
potential unit is necessary if the argument has the form $\inlex{n}$ and
otherwise none is needed. After the evaluation there
is $1$ unit left if the result is \lstinline{Some n} and
$0$ otherwise.

The right type annotation for \lstinline{filter_map} depends on the
context in which the function is used. The general type can be
described with abstract annotations and linear constraints:
%
\begin{align*}
  \mathsf{filter\_map} :\; &(\langle \mathsf{int}^{q_1} + \mathsf{bool}^{q_2}, {p_0} \rangle \to \langle \mathsf{option}^{q_3}(\mathsf{int}), p_0'\rangle) \to \\
  &\langle \mathsf{list}^{r_1}(\mathsf{int}^{r_2} + \mathsf{bool}^{r_3}), p_1 \rangle \to 
  \langle \mathsf{list}^{r_4}(\mathsf{int}), {p'_1} \rangle
\end{align*}
\[ r_1 \geq p_0 + 1, r_2 \geq q_1, r_3 \geq q_2, p_1 \geq p_1',  q_3+p_0' \geq r_4\]
To be clear, this symbolic representation cannot be expressed \emph{within} the type system. However, as part of type inference, this form is derived with the symbolic values as metavariables; the constraints are then solved using linear programming to find a solution that, when substituted in, provides a concrete judgement within the type system.
An essential requirement, then, is that the transfer of potential from the
list to its head and tail can be expressed with linear constraints.
For linear potential functions, this is straightforward
since the annotation of the head is the annotation of the element type
and the annotation of the tail is the annotation of the matched list.

\subsection{Potential functions of lists}
\label{sec:listpotential}

To go beyond linear potential, polynomial AARA extends the notation \(\langle L^{q_1} (A), q_0 \rangle\) to \(L^{(q_0, q_1, \dots, q_m)}(A)\), where \(\vec{q}\) is a vector of coefficients that specify a polynomial~\cite{hoffmann2010amortized}. What is less clear is how to maintain the aforementioned requirement for only linear constraints to come of destructing a list. The answer turns out to be a clever choice of basis: the coefficients \((q_i)\) correspond to a basis of binomial coefficients \(\binom{n}{i}\), rather than monomials \(n^i\), due to their posession of an \emph{additive shift} function \(\shift(q_0, \dots, q_m) = (q_0 + q_1, \dots, q_{m-1} + q_m, q_m)\). This is a \emph{linear} function that specifies how to preserve potential--that is, evaluating \(\shift(\vec{q})\) on \(n\) is equal to evaluating \(\vec{q}\) on \(n + 1\). This concept of a linear shift function turns out to be a key guiding abstraction that guarantees the generation of only linear constraints in the typing rule for pattern matching.

This principle carries over when AARA is extended to multivariate annotations--including terms like \(m \cdot n\), as might be required when computing the Cartesian product of two lists--but the coefficient vector notation does not. To address this, multivariate AARA introduces the use of \emph{indices} to form a basis of potential functions~\cite{hoffmann2011multivariate}. Intuitively, they generalize the notion of giving names to ``monomials'' like \(\binom{n}{2}\) or \(\binom{n}{3} \binom{m}{2}\). List indices have the form \([i_1, \dots, i_n]\), where each \(i_j\) is an index for the list elements' type. Such a list index refers to counting the number of combinations of elements of the list that match the inner indices. It's perhaps best illustrated with some examples; we'll stick with univariate examples for simplicity's sake, but it is easily extended to the multivariate case. For starters, take the index \([\star] = \star :: \mathsf{nil}\) on lists, which counts the number of ways that an element \(\star\) can be followed by \(\mathsf{nil}\), i.e., the length of the list. Visually consider evaluating it on two lists of different lengths:
\vspace{-1ex}
\newcommand{\anystar}{\star}
\tikzset{cons list/.style={start chain, node distance=1em, every on chain/.style={join=by {draw=none, "::"}}}}
\newcommand{\drawlist}[4][]{%
  \draw[cons list]
    #3
    foreach \x [count=\xi from 1] in {#4} { node[on chain,#1] (#2-\xi) {\x} }
    node[on chain] (#2-nil) {[]}
}
\newcommand{\circlehighlight}[2]{%
  \draw[#1] (#2) circle [radius=0.7em]
}
\tikzset{highlight one/.style={thick,draw=RoyalBlue}}
\tikzset{highlight two/.style={thick,draw=Orange,dashed}}
\tikzset{highlight three/.style={thick,draw=OliveGreen,dotted}}
\tikzset{highlight four/.style={thick,draw=Purple,dash dot}}
\begin{center}
\begin{tikzpicture}[scale=0.8,grow=right,
                    level distance=2em,
                    edge from parent path={(\tikzparentnode.east) [draw=none] -- node {::} (\tikzchildnode.west)},
                    every edge quotes/.style={}] 
  \node at (-1, 0) (index) {Index: };
  \drawlist{index}{($(index.east)+(1em,0)$)}{\(\anystar\)};
  \circlehighlight{highlight one}{index-1};
  \circlehighlight{highlight two}{index-nil};

  \coordinate (tablecenter) at (0,-0.7);

  \begin{scope}[shift={($(tablecenter)+(-3.6,-0.3)$)}]
    \node (value1) {Value: };
    \drawlist{value1}{($(value1.east)+(1em,0)$)}{1,2};

    \foreach \i in {1,2} {
      \drawlist[\if\i\xi\else gray\fi]{value1-\i}{($(value1-1)+(0,-1.8em * \i)$)}{1,2};
      \circlehighlight{highlight one}{value1-\i-\i};
      \circlehighlight{highlight two}{value1-\i-nil};
    }

    \draw[decoration={calligraphic brace,mirror,raise=4pt},decorate,thick]
      (value1-1-1.north west) --
      node[left=8pt,text width=5em,align=right] {\emph{Result: \textbf{2}}}
      (value1-2-1.south west);
  \end{scope}

  \begin{scope}[shift={($(tablecenter) + (0.5,-0.3)$)}]
    \drawlist{value2}{(0,0)}{1,2,3,4};

    \foreach \i in {1,2,3,4} {
      \drawlist[\if\i\xi\else gray\fi]{value2-\i}{($(value2-1)+(0,-1.8em * \i)$)}{1,2,3,4};
      \circlehighlight{highlight one}{value2-\i-\i};
      \circlehighlight{highlight two}{value2-\i-nil};
    }

    \draw[decoration={calligraphic brace,raise=6pt},decorate,thick]
      (value2-1-nil.north east) --
      node[right=10pt,text width=5em] {\emph{\textbf{4}}}
      (value2-4-nil.south east);
  \end{scope}

  \draw (value1.south west) -- (value2-nil.south east);
  \draw (tablecenter) -- (tablecenter |- value2-4-1.south west);

\end{tikzpicture}
\end{center}
\newcommand{\highone}[1]{\highlight{OliveGreen}{$\displaystyle #1$}}
\newcommand{\hightwo}[1]{\highlight{Fuchsia}{$\displaystyle #1$}}
\newcommand{\colorone}{RedOrange}
\newcommand{\co}[2]{\tikzmarknode{#1}{\textcolor{\colorone}{#2}}}
\newcommand{\colortwo}{Dandelion}
\newcommand{\ct}[2]{\tikzmarknode{#1}{\textcolor{\colortwo}{#2}}}
Note that, as demonstrated by the two circles in each evaluation, there are two matches in each: a cons cell, and the ending nil. The critical aspect of list indices' evaluation is that it can be phrased purely locally in terms of the heads and tails of the index and list elements (where \(\phi\) designates the function that evaluates an index on a value):
\vspace{1.3em}
\[\phi_{\co{i0}{i}::\ct{is0}{is}} (\highone{v} :: \hightwo{vs}) = \phi_{\co{i1}{i}} (\highone{v}) \cdot \phi_{\ct{is1}{is}} (\hightwo{vs}) + \phi_{\co{i2}{i}::\ct{is2}{is}} (\hightwo{vs})\]
\begin{tikzpicture}[overlay,remember picture]
  \draw[->,color=\colorone!50] (i0.north) + (0,0.2em) -- ++ (0,1em) -| ([yshift=0.2em] i1.north);
  \draw[->,color=\colortwo!50] (is0.north) + (0,0.2em) -- ++ (0,1.5em) -| ([yshift=0.2em] is1.north);
  \draw[->,color=\colorone!50] (i0.south) + (0,-0.1em) -- ++ (0,-1em) -| ([yshift=-0.2em] i2.south);
  \draw[->,color=\colortwo!50] (is0.south) + (0,-0.1em) -- ++ (0,-0.5em) -| ([yshift=-0.2em] is2.south);
\end{tikzpicture}

which first counts the combinations that include the head element, then adds the combinations that don't. From this presentation, an analogous shift function falls out: \(\shift(i :: is) = (i, is) + (\star, i :: is)\), where the result is evaluated on \((v, vs)\) given a list \(v :: vs\). Note just how similar this is to the definition for binomial coefficients!

As an example of how these indices are used in types, return to the second type of \lstinline{filter_map f} we presented, namely \(\langle \mathsf{list}^1(\mathsf{int}^1 + \mathsf{bool}^0), 0 \rangle \to \langle \mathsf{list}^1(int), 0 \rangle\). Expressed using indices, this function requires its argument to have potential \(2 \cdot [\inlex{\star}] + 1 \cdot [\inrex{\star}]\) and returns a value with potential \(1 \cdot [\star]\).

Building toward our desire to type \lstinline{quicksort}, first consider some evaluations of the index \([\star; \star]\):
\vspace{-1ex}
\begin{center}
\begin{tikzpicture}[scale=0.8,grow=right,
                    level distance=2em,
                    edge from parent path={(\tikzparentnode.east) [draw=none] -- node {::} (\tikzchildnode.west)},
                    every edge quotes/.style={}] 

  \node at (-1.5, 0) (index) {Index: };
  \drawlist{index}{($(index.east)+(1em,0)$)}{\(\anystar\),\(\anystar\)};
  \circlehighlight{highlight one}{index-1};
  \circlehighlight{highlight two}{index-2};
  \circlehighlight{highlight three}{index-nil};

  \coordinate (tablecenter) at (0,-0.7);

  \begin{scope}[shift={($(tablecenter)+(-3.2,-0.3)$)}]
    \node (value1) {Value: };
    \drawlist{value1}{($(value1.east)+(1em,0)$)}{1,2};

    \drawlist{value1-1}{($(value1-1)+(0,-1.8em)$)}{1,2};
    \circlehighlight{highlight one}{value1-1-1};
    \circlehighlight{highlight two}{value1-1-2};
    \circlehighlight{highlight three}{value1-1-nil};

    \draw[decoration={calligraphic brace,mirror,raise=4pt},decorate,thick]
      (value1-1-1.north west) --
      node[left=8pt,align=right] {\emph{Result: \textbf{1}}}
      (value1-1-1.south west);
  \end{scope}

  \begin{scope}[shift={($(tablecenter) + (0.7,-0.3)$)}]
    \drawlist{value2}{(0,0)}{1,2,3,4};

    \foreach \i / \j [count=\n] in {1/2,1/3,1/4,2/3,2/4,3/4} {
      \drawlist[\if\i\xi\else\if\j\xi\else gray\fi\fi]
               {value2-\n}{($(value2-1)+(0,-1.8em * \n)$)}{1,2,3,4};
      \circlehighlight{highlight one}{value2-\n-\i};
      \circlehighlight{highlight two}{value2-\n-\j};
      \circlehighlight{highlight three}{value2-\n-nil};
    }

    \draw[decoration={calligraphic brace,raise=6pt},decorate,thick]
      (value2-1-nil.north east) --
      node[right=10pt,text width=5em] {\emph{\textbf{6}}}
      (value2-6-nil.south east);
  \end{scope}

  \draw (value1.south west) -- (value2-nil.south east);
  \coordinate (centerline) at ($(tablecenter) + (0.3,0)$);
  \draw (centerline) -- (centerline |- value2-6-1.south west);

  \useasboundingbox (-0.5\linewidth,0) rectangle (0.5\linewidth,0);
\end{tikzpicture}
\end{center}
As expected, we find that this index corresponds to 
\(\binom{n}{2}\). Thus, given that we 
know \lstinline{quicksort} has cost \(n^2\), 
we can express the required potential of its argument using indices as \(2 \cdot [\star; \star] + 1 \cdot [\star]\). Finally, we can consider our original function, \lstinline{sort_lefts_list}. Here we can see that it must require an input potential of \(2 \cdot [\inlex{\star}; \inlex{\star}] + 2 \cdot [\inlex{\star}]\)--\lstinline{filter_map} consumes the \(1 \cdot [\inlex{\star}]\) part of it and passes on the rest to \lstinline{quicksort}.


\subsection{Extending to regular recursive types}
\label{sec:nestedpotential}

However, these indices do not obviously generalize to regular inductive types. Jost et al.~\cite{jost2010static} handle potential on regular inductives, but only in the very restricted setting of univariate linear potential, which amounts to just counting constructors. Hoffmann et al.~\cite{hoffmann2011multivariate} and their successor works handle more expressive potential functions, but don't support regular inductives and treat even just binary trees as lists for potential purposes. Tree indices are identical to list indices, and tree values are just list versions of themselves flattened by a preorder traversal. This results in the combinatorial \emph{structure} of trees being completely lost.

Let's explore a different design. For one, we know we absolutely must preserve some sort of linear shift function. Another hint comes from Hoffmann et al.~\cite{hoffmann2011multivariate}, who observe in passing that their indices for a type \(\tau\) essentially follow the structure of values of type \(\tau\). We find that they were on to something after all. We consider indices that correspond almost \emph{exactly} to the values of the type they describe. To build intuition, we'll first give some examples on specific data types before we get to describing the general case.

\subsubsection{Stepping stone: binary trees}

We'll start by looking at the case of binary trees. In the following diagrams, tree nodes are circles while leaves are triangles. Consider evaluating the ``leaf'' index on two different trees:
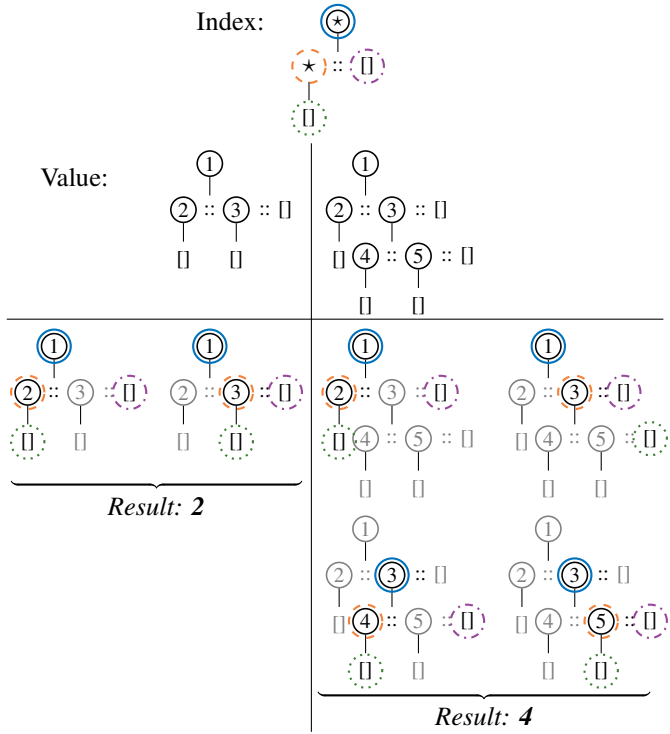
\tikzset{treenode/.style={draw,semithick,circle,inner sep=1pt,font=\footnotesize}}
\tikzset{treeleaf/.style={draw,semithick,regular polygon,regular polygon sides=3,inner sep=1.5pt}}
\begin{center}
\begin{tikzpicture}[scale=0.9,level distance=1.2em,sibling distance=2em, node distance=1em]
  \node at (-0.6, 0) (index) {Index: };
  \node[treeleaf] (index-1) [right=of index] {};
  \circlehighlight{highlight one}{index-1};

  \coordinate (tablecenter) at (0,-0.5);

  \node[treeleaf] at (0.3\textwidth,0) (legendleaf) {};
  \node at ($(0.3\textwidth - 0.7cm, 0.03cm)$) {leaf: };
  \node[treenode] at (0.3\textwidth,-5mm) {\phantom{1}};
  \node at ($(0.3\textwidth - 0.8cm, -4.8mm)$) {node: };

  \begin{scope}[shift={($(tablecenter)+(-1.0,-0.3)$)}]
    \node at (-2,-0.2) (value1label) {Value:};
    \foreach \coord / \rootcolor / \onecolor / \twocolor [count=\i]
             in {(0,0)///,(-1.5,-1.5)/gray//gray,(0,-1.5)/gray/gray/} {
      \node[treenode,\rootcolor] at \coord {1}
        child { node[treeleaf,\onecolor] (value1-\i-1) {} }
        child { node[treeleaf,\twocolor] (value1-\i-2) {} };
    }

    \circlehighlight{highlight one}{value1-2-1};
    \circlehighlight{highlight one}{value1-3-2};

    \draw[decoration={calligraphic brace,mirror,raise=8pt},decorate,thick]
      (value1-2-1.south west) --
      node[below=10pt] {\emph{Result: \textbf{2}}}
      (value1-3-2.south east);
  \end{scope}

  \begin{scope}[shift={($(tablecenter)+(1.0,-0.3)$)}]
    \foreach \coord / \rootcolor / \onecolor / \intercolor / \twocolor / \threecolor [count=\i] 
             in {(0,0)/////,(0,-1.5)/gray//gray/gray/gray,(1.5,-1.5)/gray/gray/gray//gray,(3,-1.5)/gray/gray/gray/gray/} {
      \node[treenode,\rootcolor] at \coord (value2-\i-root) {1}
        child { node[treeleaf,\onecolor] (value2-\i-1) {} }
        child { node[treenode,\intercolor] {2}
                  child { node[treeleaf,\twocolor] (value2-\i-2) {} }
                  child { node[treeleaf,\threecolor] (value2-\i-3) {} }
              };
    }
    \circlehighlight{highlight one}{value2-2-1};
    \circlehighlight{highlight one}{value2-3-2};
    \circlehighlight{highlight one}{value2-4-3};

    \draw[decoration={calligraphic brace,mirror,raise=8pt},decorate,thick]
      (value2-2-1.west |- value2-4-3.south) --
      node[below=10pt] (value2eval) {\emph{Result: \textbf{3}}}
      (value2-4-3.south east);
  \end{scope}

  \coordinate (midway) at ($(value2-1-3.south)!.5!(value2-2-root.north)$);
  \draw (value1label.west |- midway) -- (value2-4-3.east |- midway);
  \draw (tablecenter) -- (tablecenter |- value2eval.south);

\end{tikzpicture}
\end{center}
\vspace{-0.5em}
This counts the number of leaves in the tree, just as the first list index example (consisting of a cons node) counted the number of cons nodes in a list. Now let's look at the next simplest index, a node connecting two leaves:
\begin{center}
\begin{tikzpicture}[scale=0.9,level distance=1.4em,sibling distance=2.2em, node distance=1em]
  \node at (-1, 0) (index) {Index: };
  \node[treenode] (index-root) [right=2em of index] {\(\anystar\)}
    child { node[treeleaf] (index-1) {} }
    child { node[treeleaf] (index-2) {} };
  \circlehighlight{highlight one}{index-root};
  \circlehighlight{highlight two}{index-1};
  \circlehighlight{highlight three}{index-2};

  \coordinate (tablecenter) at (0,-0.9);

  \begin{scope}[shift={($(tablecenter)+(-1.0,-0.3)$)}]
    \node at (-2,-0.2) (value1label) {Value:};
    \foreach \coord / \nodecolor [count=\i] in {(0,0)/,(0,-1.7)/} {
      \node[treenode,\nodecolor] at \coord (value1-\i-root) {1}
        child { node[treeleaf] (value1-\i-1) {} }
        child { node[treeleaf] (value1-\i-2) {} };
    }

    \circlehighlight{highlight one}{value1-2-root};
    \circlehighlight{highlight two}{value1-2-1};
    \circlehighlight{highlight three}{value1-2-2};

    \draw[decoration={calligraphic brace,mirror,raise=8pt},decorate,thick]
      (value1-2-1.south west) --
      node[below=10pt,text width=5em,align=center] {\emph{Result: \textbf{1}}}
      (value1-2-2.south east);
  \end{scope}

  \begin{scope}[shift={($(tablecenter)+(1.0,-0.3)$)}]
    \foreach \coord / \rootcolor / \onecolor / \intercolor / \twocolor / \threecolor [count=\i] 
             in {(0,0)/////,(0,-1.7)///gray//gray,(1.7,-1.7)///gray/gray/,(3.4,-1.7)/gray/gray///} {
      \node[treenode,\rootcolor] at \coord (value2-\i-root) {1}
        child { node[treeleaf,\onecolor] (value2-\i-1) {} }
        child { node[treenode,\intercolor] (value2-\i-inter) {2}
                  child { node[treeleaf,\twocolor] (value2-\i-2) {} }
                  child { node[treeleaf,\threecolor] (value2-\i-3) {} }
              };
    }

    \circlehighlight{highlight one}{value2-2-root};
    \circlehighlight{highlight two}{value2-2-1};
    \circlehighlight{highlight three}{value2-2-2};

    \circlehighlight{highlight one}{value2-3-root};
    \circlehighlight{highlight two}{value2-3-1};
    \circlehighlight{highlight three}{value2-3-3};

    \circlehighlight{highlight one}{value2-4-inter};
    \circlehighlight{highlight two}{value2-4-2};
    \circlehighlight{highlight three}{value2-4-3};

    \draw[decoration={calligraphic brace,mirror,raise=8pt},decorate,thick]
      (value2-2-1.west |- value2-4-3.south) --
      node[below=10pt] (value2eval) {\emph{Result: \textbf{3}}}
      (value2-4-3.south east);
  \end{scope}

  \coordinate (midway) at ($(value2-1-3.south)!.5!(value2-2-root.north)$);
  \draw (value1label.west |- midway) -- (value2-4-3.east |- midway);
  \draw (tablecenter) -- (tablecenter |- value2eval.south);
\end{tikzpicture}
\end{center}
\vspace{-0.5em}
The evaluation on the right may be confusing at first--isn't there only one subtree that matches the index? The answer may be seen in analogy with the combinatorial evaluation on lists presented earlier: all possible combinations of constructors are considered, subject to the ordering imposed by the index.

These examples are instances of the rules for binary trees, again defined purely locally:
\newcommand{\treeleaf}{\tikz[scale=0.8] \node[treeleaf] {};}
\newcommand{\treenode}[3]{%
  \tikz[level distance=1em,sibling distance=1.5em,baseline=($(root.south)!.5!(leaf1)$)]
    \node[treenode] (root) {\(#1\)}
      child { node[font=\footnotesize,inner sep=0em] (leaf1) {\(#2\)} }
      child { node[font=\footnotesize,inner sep=0em] (leaf2) {\(#3\)} };
}
\vspace{-1em}
\begin{align*}
  \phi_{\treeleaf}(\treeleaf) &= 1 \\
  \phi_{\treenode{i}{i_1}{i_2}}(\treeleaf) &= 0 \\[0.2em]
  \phi_{\treeleaf}\left( \treenode{v}{t_1}{t_2} \right) &= \phi_{\treeleaf}(t_1) + \phi_{\treeleaf}(t_2) \\[0.2em]
  \phi_{\treenode{i}{i_1}{i_2}}\left( \treenode{v}{t_1}{t_2} \right) &= \phi_i(v) \cdot \phi_{i_1}(t_1) \cdot \phi_{i_2}(t_2) \\[-0.8em]
  &+ \phi_{\treenode{i}{i_1}{i_2}}(t_1) + \phi_{\treenode{i}{i_1}{i_2}}(t_2)
\end{align*}

\subsubsection{Eureka!}

The insight for the general case, then, is to notice this correspondence between the rules for lists and the rules for binary trees:
\begin{align*}
  \phi_{i :: is} (v :: vs) &= \tikzmarknode{listhere}{\phi_i(v) \cdot \phi_{is}(vs)} \\[0.3em]
  &+ \tikzmarknode{listthere}{\phi_{i :: is}(vs)} \\[1em]
  \phi_{\treenode{i}{i_1}{i_2}}\left( \treenode{v}{t_1}{t_2} \right) &= \tikzmarknode{treehere}{\phi_i(v) \cdot \phi_{i_1}(t_1) \cdot \phi_{i_2}(t_2)} \\
  &+ \tikzmarknode{treethere}{\phi_{\treenode{i}{i_1}{i_2}}(t_1) + \phi_{\treenode{i}{i_1}{i_2}}(t_2)}
\end{align*}
\begin{tikzpicture}[overlay,remember picture]
  \draw[thick,dashed,red] ($(listhere.north west)+(-1pt,1pt)$) rectangle ($(listhere.south east)+(1pt,-1pt)$);
  \draw[thick,dotted,blue] ($(listthere.north west)+(-1pt,1pt)$) rectangle ($(listthere.south east)+(1pt,-1pt)$);
  \draw[thick,dashed,red] ($(treehere.north west)+(-1pt,1pt)$) rectangle ($(treehere.south east)+(1pt,-1pt)$);
  \draw[thick,dotted,blue] ($(treehere.south west)+(-2pt,-9pt)$) rectangle ($(treehere.south east)+(6pt,-36pt)$);
\end{tikzpicture}

\begin{intuition*}
  To evaluate an index at a constructor, first evaluate it at the immediate constructor, then add that to the evaluation of the original index at all direct children.
\end{intuition*}
Note that it satisfies our desired properties: it is multivariate through the use of multiplication at the immediate constructor evaluation; it is structure-dependent by evaluating recursively only at direct children; and, critically, it suggests a shift function that exactly mirrors this construction. Having observed that, we will leave it to \cref{sec:indexeval} to define this formally, but we will at least address one ambiguity in that specification: what are the ``direct children'' of a constructor?

\subsubsection{The prize: rose trees}
\label{sec:rosetrees}

\newcommand{\bnode}{\phantom{1}}

The direct children of a cons cell or tree node are readily apparent, but they are less obvious for our original motivating data type, the rose tree. Let us again turn to examples, starting with the simplest index:
\begin{center}
\begin{tikzpicture}[level distance=1.4em,sibling distance=2.2em, node distance=0.7em,scale=0.9]
  \node at (-1, 0) (index) {Index: };
  \node[treenode] (index-root) [right=2em of index] {\(\star\)};
  \node[below=of index-root,font=\footnotesize] (index-leaf) {[]};
  \draw (index-root) -- (index-leaf);
  \circlehighlight{highlight one}{index-root};
  \circlehighlight{highlight two}{index-leaf};

  \coordinate (tablecenter) at (0,-1.2);

  \begin{scope}[shift={($(tablecenter)+(-1.3,-0.3)$)},font=\footnotesize]
    \node[font=\normalsize] at (-2,-0.2) (value1label) {Value:};
    \foreach \coord / \ca / \cb / \cc / \cd / \ce / \cf [count=\i] in {(0,0)//////,(0,-2.7)//gray/gray/gray/gray/,(-2.2,-5.4)/gray///gray/gray/gray,(0,-5.4)/gray/gray/gray///gray} {
      \node[treenode,\ca] at \coord (value1-\i-root) {1};
      \node[\cb,below=of value1-\i-root] (value1-\i-cons1) {::};
      \draw (value1-\i-root) -- (value1-\i-cons1);

      \node[\cb,treenode,left=-0.3mm of value1-\i-cons1] (value1-\i-2) {2};
      \node[\cc,below=of value1-\i-2] (value1-\i-2nil) {[]};
      \draw (value1-\i-2) -- (value1-\i-2nil);

      \node[\cd,treenode,right=-0.3mm of value1-\i-cons1] (value1-\i-3) {3};
      \node[\cd,right=0mm of value1-\i-3] (value1-\i-cons2) {::};
      \node[\ce,below=of value1-\i-3] (value1-\i-3nil) {[]};
      \draw (value1-\i-3) -- (value1-\i-3nil);

      \node[\cf,right=-1.5mm of value1-\i-cons2] (value1-\i-nil) {[]};
    }

    \circlehighlight{highlight one}{value1-2-root};
    \circlehighlight{highlight two}{value1-2-nil};

    \circlehighlight{highlight one}{value1-3-2};
    \circlehighlight{highlight two}{value1-3-2nil};

    \circlehighlight{highlight one}{value1-4-3};
    \circlehighlight{highlight two}{value1-4-3nil};

    \draw[decoration={calligraphic brace,mirror,raise=8pt},decorate,thick]
      (value1-3-2nil.south west) --
      node[below=10pt,font=\normalfont,align=center] {\emph{Result: \textbf{3}}}
      (value1-3-2nil.south west -| value1-4-nil.south east);
  \end{scope}

  \begin{scope}[shift={($(tablecenter)+(1.0,-0.3)$)},font=\footnotesize]
    \foreach \coord / \ca / \cb / \cc / \cd / \ce / \cf / \cg / \ch / \ci / \cj [count=\i] in {(0,0)//////////,(0,-2.7)//gray/gray/gray/gray/gray/gray/gray/gray/,(0,-5.4)/gray///gray/gray/gray/gray/gray/gray/gray,(2.4,-5.4)/gray/gray/gray//gray/gray/gray/gray//gray,(0,-8.1)/gray/gray/gray/gray///gray/gray/gray/gray,(2.4,-8.1)/gray/gray/gray/gray/gray/gray///gray/gray} {
      \node[treenode,\ca] at \coord (value2-\i-root) {1};
      \node[\cb,below=of value2-\i-root] (value2-\i-cons1) {::};
      \draw (value2-\i-root) -- (value2-\i-cons1);

      \node[\cb,treenode,left=-0.3mm of value2-\i-cons1] (value2-\i-2) {2};
      \node[\cc,below=of value2-\i-2] (value2-\i-2nil) {[]};
      \draw (value2-\i-2) -- (value2-\i-2nil);

      \node[\cd,treenode,right=-0.3mm of value2-\i-cons1] (value2-\i-3) {3};
      \node[\cd,right=0mm of value2-\i-3] (value2-\i-cons2) {::};
      \node[\ce,below=of value2-\i-3] (value2-\i-3cons) {::};
      \draw (value2-\i-3) -- (value2-\i-3cons);

      \node[\ce,treenode,left=-0.3mm of value2-\i-3cons] (value2-\i-4) {4};
      \node[\cf,below=of value2-\i-4] (value2-\i-4nil) {[]};
      \draw (value2-\i-4) -- (value2-\i-4nil);

      \node[\cg,treenode,right=-0.3mm of value2-\i-3cons] (value2-\i-5) {5};
      \node[\cg,right=0mm of value2-\i-5] (value2-\i-cons3) {::};
      \node[\ch,below=of value2-\i-5] (value2-\i-5nil) {[]};
      \draw (value2-\i-5) -- (value2-\i-5nil);
      \node[\ci,right=-1.5mm of value2-\i-cons3] (value2-\i-nil2) {[]};

      \node[\cj,right=-1.5mm of value2-\i-cons2] (value2-\i-nil) {[]};
    }

    \circlehighlight{highlight one}{value2-2-root};
    \circlehighlight{highlight two}{value2-2-nil};

    \circlehighlight{highlight one}{value2-3-2};
    \circlehighlight{highlight two}{value2-3-2nil};

    \circlehighlight{highlight one}{value2-4-3};
    \circlehighlight{highlight two}{value2-4-nil2};

    \circlehighlight{highlight one}{value2-5-4};
    \circlehighlight{highlight two}{value2-5-4nil};

    \circlehighlight{highlight one}{value2-6-5};
    \circlehighlight{highlight two}{value2-6-5nil};

    \draw[decoration={calligraphic brace,mirror,raise=8pt},decorate,thick]
      (-0.7, -10.3) --
      node[below=10pt,font=\normalfont,align=center] {\emph{Result: \textbf{5}}}
      (4, -10.3);
  \end{scope}

  \draw (-4.0,-3.8) -- (5,-3.8);
  \draw (0.2,-1) -- (0.2,-12.7);

\end{tikzpicture}
\end{center}

And a more complex index:
\begin{center}
\begin{tikzpicture}[level distance=1.4em,sibling distance=2.2em, node distance=0.7em,scale=0.9]
  \node at (-1, 0) (index) {Index: };
  \node[treenode] (index-root) [right=2em of index] {\(\star\)};
  \node[below=of index-root,font=\footnotesize] (index-cons) {::};
  \node[left=-0.3mm of index-cons] (index-node) {\(\star\)};
  \draw (index-root) -- (index-cons);
  \node[right=-0.3mm of index-cons,font=\footnotesize] (index-nil) {[]};
  \node[below=of index-node,font=\footnotesize] (index-leaf) {[]};
  \draw (index-node) -- (index-leaf);

  \circlehighlight{highlight one}{index-root};
  \circlehighlight{highlight two}{index-node};
  \circlehighlight{highlight three}{index-leaf};
  \circlehighlight{highlight four}{index-nil};

  \coordinate (tablecenter) at (0,-1.8);

  \begin{scope}[shift={($(tablecenter)+(-1.3,-0.3)$)},font=\footnotesize]
    \node[font=\normalsize] at (-2,-0.2) (value1label) {Value:};
    \foreach \coord / \ca / \cb / \cc / \cd / \ce / \cf [count=\i] in {(0,0)//////,(-2.3,-2.7)////gray/gray/,(0,-2.7)//gray/gray///} {
      \node[treenode,\ca] at \coord (value1-\i-root) {1};
      \node[\cb,below=of value1-\i-root] (value1-\i-cons1) {::};
      \draw (value1-\i-root) -- (value1-\i-cons1);

      \node[\cb,treenode,left=-0.3mm of value1-\i-cons1] (value1-\i-2) {2};
      \node[\cc,below=of value1-\i-2] (value1-\i-2nil) {[]};
      \draw (value1-\i-2) -- (value1-\i-2nil);

      \node[\cd,treenode,right=-0.3mm of value1-\i-cons1] (value1-\i-3) {3};
      \node[\cd,right=0mm of value1-\i-3] (value1-\i-cons2) {::};
      \node[\ce,below=of value1-\i-3] (value1-\i-3nil) {[]};
      \draw (value1-\i-3) -- (value1-\i-3nil);

      \node[\cf,right=-1.5mm of value1-\i-cons2] (value1-\i-nil) {[]};
    }

    \circlehighlight{highlight one}{value1-2-root};
    \circlehighlight{highlight two}{value1-2-2};
    \circlehighlight{highlight three}{value1-2-2nil};
    \circlehighlight{highlight four}{value1-2-nil};

    \circlehighlight{highlight one}{value1-3-root};
    \circlehighlight{highlight two}{value1-3-3};
    \circlehighlight{highlight three}{value1-3-3nil};
    \circlehighlight{highlight four}{value1-3-nil};

    \draw[decoration={calligraphic brace,mirror,raise=8pt},decorate,thick]
      (value1-2-2nil.south west) --
      node[below=10pt,font=\normalfont,align=center] {\emph{Result: \textbf{2}}}
      (value1-2-2nil.south west -| value1-3-nil.south east);
  \end{scope}

  \begin{scope}[shift={($(tablecenter)+(1.0,-0.3)$)},font=\footnotesize]
    \foreach \coord / \ca / \cb / \cc / \cd / \ce / \cf / \cg / \ch / \ci / \cj [count=\i] in {(0,0)//////////,(0,-2.7)////gray/gray/gray/gray/gray/gray/,(2.7,-2.7)//gray/gray//gray/gray/gray/gray//,(0,-5.4)/gray/gray/gray////gray/gray//gray,(2.7,-5.4)/gray/gray/gray//gray/gray////gray} {
      \node[treenode,\ca] at \coord (value2-\i-root) {1};
      \node[\cb,below=of value2-\i-root] (value2-\i-cons1) {::};
      \draw (value2-\i-root) -- (value2-\i-cons1);

      \node[\cb,treenode,left=-0.3mm of value2-\i-cons1] (value2-\i-2) {2};
      \node[\cc,below=of value2-\i-2] (value2-\i-2nil) {[]};
      \draw (value2-\i-2) -- (value2-\i-2nil);

      \node[\cd,treenode,right=-0.3mm of value2-\i-cons1] (value2-\i-3) {3};
      \node[\cd,right=0mm of value2-\i-3] (value2-\i-cons2) {::};
      \node[\ce,below=of value2-\i-3] (value2-\i-3cons) {::};
      \draw (value2-\i-3) -- (value2-\i-3cons);

      \node[\ce,treenode,left=-0.3mm of value2-\i-3cons] (value2-\i-4) {4};
      \node[\cf,below=of value2-\i-4] (value2-\i-4nil) {[]};
      \draw (value2-\i-4) -- (value2-\i-4nil);

      \node[\cg,treenode,right=-0.3mm of value2-\i-3cons] (value2-\i-5) {5};
      \node[\cg,right=0mm of value2-\i-5] (value2-\i-cons3) {::};
      \node[\ch,below=of value2-\i-5] (value2-\i-5nil) {[]};
      \draw (value2-\i-5) -- (value2-\i-5nil);
      \node[\ci,right=-1.5mm of value2-\i-cons3] (value2-\i-nil2) {[]};

      \node[\cj,right=-1.5mm of value2-\i-cons2] (value2-\i-nil) {[]};
    }

    \circlehighlight{highlight one}{value2-2-root};
    \circlehighlight{highlight two}{value2-2-2};
    \circlehighlight{highlight three}{value2-2-2nil};
    \circlehighlight{highlight four}{value2-2-nil};

    \circlehighlight{highlight one}{value2-3-root};
    \circlehighlight{highlight two}{value2-3-3};
    \circlehighlight{highlight three}{value2-3-nil2};
    \circlehighlight{highlight four}{value2-3-nil};

    \circlehighlight{highlight one}{value2-4-3};
    \circlehighlight{highlight two}{value2-4-4};
    \circlehighlight{highlight three}{value2-4-4nil};
    \circlehighlight{highlight four}{value2-4-nil2};

    \circlehighlight{highlight one}{value2-5-3};
    \circlehighlight{highlight two}{value2-5-5};
    \circlehighlight{highlight three}{value2-5-5nil};
    \circlehighlight{highlight four}{value2-5-nil2};

    \draw[decoration={calligraphic brace,mirror,raise=8pt},decorate,thick]
      (-0.7, -7.5) --
      node[below=10pt,font=\normalfont,align=center] {\emph{Result: \textbf{4}}}
      (4.2, -7.5);
  \end{scope}

  \draw (-4.3,-4.4) -- (5.4,-4.4);
  \draw (0.2,-1.8) -- (0.2,-10.5);

\end{tikzpicture}
\end{center}
To specify the direct children of a rose tree node, we piggyback off of the list's notion of direct children: a rose tree node's direct child is any node that appears in its list. This notion of pushing the problem of recursive evaluation of the outer type down to the inner type is precisely the solution. Speaking anthropomorphically, the rose tree can identify, in any given list node, the one possible occurrence of a tree (in a cons cell); the list can then use that information to look through the recursive occurences of the list. This intuition is formalized and explained once again in~\cref{sec:respolydefs}.

Calling back to our motivating \lstinline{filter_map_tree}, specifying a required potential for the same function as the second typing of \lstinline{filter_map} is now as simple as \(2 \cdot \mathsf{Tree}(\inlex{\star}, []) + 1 \cdot \mathsf{Tree}(\inrex{\star}, [])\), i.e., \(2m + 1n\) where \(m\) is the number of nodes with ints and \(n\) is the number of nodes with bools. For the overall \lstinline{sort_lefts_tree}, it is the similarly natural \(2 \cdot \mathsf{Tree}(\inlex{\star}, [\mathsf{Tree}(\inlex{\star}, [])]) + 2 \cdot \mathsf{Tree}(\inlex{\star}, [])\), for much the same reasons as the list case. Incredibly, these indices look nearly as simple as the indices for the equivalent list functions, which we believe is a strong suggestion of elegance.

\section{Resource Polynomials}
\label{sec:respolys}

As in previous AARA type systems, \emph{resource polynomials} serve as our language's mechanism to assign potential to typed values. Our core contribution to their theory is a generalization of past systems' bounded-branching tree types to more general algebraic, possibly-mutually recursive types. In this section, we first formally define these potential functions, then give manipulations of them necessary for the type system, continuing our use of running examples to illustrate the definitions.

\subsection{Resource polynomial definitions}
\label{sec:respolydefs}

\subsubsection*{Types and values}
\label{sec:typesandvalues}

To show to what exactly resource polynomials assign potential, we first give the types and values over which the resource polynomials are defined in Figure~\ref{fig:typesandvalues}. The types presented are standard, save the arrow type--the details of which are irrelevant to the resource polynomials and explained in \cref{sec:typesystem}. We also give, in Figure~\ref{fig:typesandvalues}, inference rules for the set of syntactically valid values \(\mathcal{V}(\tau)\) for a given type \(\tau\). The notation \([\sigma/\alpha] \tau\) refers to the capture-avoiding substitution of \(\sigma\) for \(\alpha\) in \(\tau\). Note that these typing rules do not guarantee anything for the purposes of language semantics--in particular, the function type here is practically unrestricted--but are instead just to guarantee that the potential function can be evaluated on indices and values of matching types.

\begin{figure}
  \centering
  \begin{subfigure}[b]{0.95\linewidth}
    \centering
    \[\renewcommand{\arraystretch}{1.1}
      \begin{array}{rrclr}
      \textrm{Types} & \tau &\bnfdef &\alpha &\textrm{Type variable} \\
                    &&\bnfalt &\mathbf{1} &\textrm{Unit} \\
                    &&\bnfalt &\tau_1 \times \tau_2 &\textrm{Product} \\
                    &&\bnfalt &\tau_1 + \tau_2 &\textrm{Sum} \\
                    &&\bnfalt &\funty{\tau_1}{\tau_2}{\Theta}{\Theta_{\costfree}} &\textrm{Arrow} \\
                    &&\bnfalt &\recty{\alpha}{\tau} &\textrm{Isorecursive}
      \end{array}\]
  \end{subfigure}
  \begin{subfigure}[b]{0.95\linewidth}
    \centering
    \[\renewcommand{\arraystretch}{1.1}
      \begin{array}{rrclcl}
      \textrm{Values} & v
       &\bnfdef &\unitex &\bnfalt &\pairex{v_1}{v_2} \\
      &&\bnfalt &\inlex{v} &\bnfalt &\inrex{v} \\
      &&\bnfalt &\funex{f}{x}{e} &\bnfalt &\foldex{v}
      \end{array}\]
  \end{subfigure}
  \begin{subfigure}[T]{\linewidth}
    \centering
    \begin{mathpar}
      \inferrule
      {\strut}
      {\unitex \in \mathcal{V}(\mathbf{1})}

      \inferrule
      {\strut}
      {\funex{f}{x}{e} \in \mathcal{V}(\funty{\tau_1}{\tau_2}{\Theta}{\Theta_{\costfree}})}

      \inferrule
      {v_1 \in \mathcal{V}(\tau_1) \\ v_2 \in \mathcal{V}(\tau_2)}
      {\pairex{v_1}{v_2} \in \mathcal{V}(\tau_1 \times \tau_2)}

      \inferrule
      {v_1 \in \mathcal{V}(\tau_1)}
      {\inlex{v_1} \in \mathcal{V}(\tau_1 + \tau_2)}

      \inferrule
      {v_2 \in \mathcal{V}(\tau_2)}
      {\inrex{v_2} \in \mathcal{V}(\tau_1 + \tau_2)}

      \inferrule
      {v \in \mathcal{V}([\recty{\alpha}{\tau} / \alpha]\tau)}
      {\foldex{v} \in \mathcal{V}(\recty{\alpha}{\tau})}
    \end{mathpar}
  \end{subfigure}
  \caption{Types, values, and typing of values}\label{fig:typesandvalues}
\end{figure}

Following our running examples, we may define the types \(\mathsf{bool} \triangleq \sumty{\unitty}{\unitty}\), \(\mathsf{list}(\tau) \triangleq \recty{\alpha}{\sumty{\unitty}{\prodty{\tau}{\alpha}}}\), and \(\mathsf{tree}(\tau) \triangleq \recty{\beta}{\prodty{\tau}{\mathsf{list}(\beta)}} = \recty{\beta}{\prodty{\tau}{(\recty{\alpha}{\sumty{\unitty}{\prodty{\beta}{\alpha}}})}}\), with value constructors \(\mathsf{True} \triangleq \inlex{\unitex}\) and \(\mathsf{False} \triangleq \inrex{\unitex}\), \(\mathsf{Nil} \triangleq \foldex{(\inlex{\unitex})}\) and \(\mathsf{Cons}(h, t) \triangleq \foldex{(\inrex{(\pairex{h}{t})})}\), and \(\mathsf{Tree}(x, t) \triangleq \foldex{(\pairex{x}{t})}\), respectively. We use the notation \([v_1, \dots, v_n]\) to refer to \(\mathsf{Cons}(v_1, \dots (\mathsf{Cons}(v_n, \mathsf{Nil})))\).

\subsubsection*{Indices}
\label{sec:indices}

\begin{figure*}
  \centering
  \begin{subfigure}[t]{0.48\textwidth}
    \judgement{Base polynomial indices}{\(i \in \mathcal{I}(\tau)\)}
    \vfill
    \begin{mathpar}
      \inferrule
      {\strut}
      {\unitex \in \mathcal{I}(\mathbf{1})}

      \inferrule
      {i_1 \in \mathcal{I}(\tau_1) \\ i_2 \in \mathcal{I}(\tau_2)}
      {\pairex{i_1}{i_2} \in \mathcal{I}(\tau_1 \times \tau_2)}

      \inferrule
      {i_1 \in \mathcal{I}(\tau_1)}
      {\inlex{i_1} \in \mathcal{I}(\tau_1 + \tau_2)}

      \inferrule
      {i_2 \in \mathcal{I}(\tau_2)}
      {\inlex{i_2} \in \mathcal{I}(\tau_1 + \tau_2)}

      \inferrule
      {\strut}
      {\constidx \in \mathcal{I}(\funty{\tau_1}{\tau_2}{\Theta}{\Theta_{\costfree}})}

      \inferrule
      {i \in \mathcal{I}([\recty{\alpha}{\tau} / \alpha]\tau)}
      {\foldex{i} \in \mathcal{I}(\recty{\alpha}{\tau})}

      \inferrule
      {\strut}
      {\lastidx \in \mathcal{I}(\recty{\alpha}{\tau})}
    \end{mathpar}
    \vfill
  \end{subfigure}
  \hfill
  \begin{subfigure}[t]{0.48\textwidth}
    \judgement{Recursive occurrence indices}{\(\mkInd{\alpha}{\tau}{i}\)}
    \centering
    \begin{align*}
      \mkInd{\alpha}{\alpha}{i} &= \{i\} \\
      \mkInd{\alpha}{\topty}{i} &= \emptyset \\
      \mkInd{\alpha}{1}{i} &= \emptyset \\
      \mkInd{\alpha}{\tau_1 + \tau_2}{i} &= \{\inlex{j} \,|\, j \in \mkInd{\alpha}{\tau_1}{i}\} \mathop{\cup} \\
                                &\mathrel{\phantom{=}} \{\inrex{j} \,|\, j \in \mkInd{\alpha}{\tau_2}{i}\} \\
      \mkInd{\alpha}{\tau_1 \times \tau_2}{i} &= \{\pairex{j}{c} \,|\, j \in \mkInd{\alpha}{\tau_1}{i}, \\
      &\hspace{6.07em} c \in \constidxs{\tau_2}\} \mathop{\cup} \\ 
                                &\mathrel{\phantom{=}} \{\pairex{c}{j} \,|\, c \in \constidxs{\tau_1}, \\
      &\hspace{6.07em} j \in \mkInd{\alpha}{\tau_2}{i}\} \\
      \mkInd{\alpha}{\tau_1 \to \tau_2}{i} &= \emptyset \\
      \mkInd{\alpha}{\recty{\beta}{\tau}}{i} &= \{\foldex{j} \,|\, j \in \mkInd{\alpha}{[\topty/\beta]\tau}{i}\}
    \end{align*}
  \end{subfigure}
  \\
  \vspace{1em}
  \begin{subfigure}[b]{0.34\textwidth}
    \judgement{Constant index set}{\(\mathcal{C}(\tau)\)}
    \vspace{1.7em}
    \centering
    \begin{align*}
      \constidxs{\mathbf{1}} &= \{\unitex\} \\
      \constidxs{\tau_1 \times \tau_2} &= \{\pairex{i_1}{i_2} \,|\, i_1 \in \constidxs{\tau_1}, \\
                  &\hspace{6.55em} i_2 \in \constidxs{\tau_2}\} \\ 
      \constidxs{\tau_1 + \tau_2} &= \{\inlex{i_1} \,|\, i_1 \in \constidxs{\tau_1}\} \mathop{\cup} \\
                  &\mathrel{\phantom{=}} \{\inrex{i_2} \,|\, i_2 \in \constidxs{\tau_2}\} \\
      \constidxs{\tau_1 \to \tau_2} &= \{\constidx\} \\
      \constidxs{\recty{\alpha}{\tau}} &= \{\lastidx\} \\
      \constidxs{\topty} &= \{\lastidx\}
    \end{align*}
    \vspace{.7em}
  \end{subfigure}
  \hfill
  \begin{subfigure}[b]{0.62\textwidth}
    \judgement{Base polynomial evaluation}{\(\pot{\tau}{i}{v}\)}
    \begin{align*}
      \pot{\mathbf{1}}{\unitex}{\unitex} &= 1 \\
      \pot{\tau_1 \times \tau_2}{\pairex{i_1}{i_2}}{\pairex{v_1}{v_2}} &= \pot{\tau_1}{i_1}{v_1} \cdot \pot{\tau_2}{i_2}{v_2} \\
      \pot{\tau_1 + \tau_2}{\inlex{i_1}}{\inlex{v_1}} &= \pot{\tau_1}{i_1}{v_1} \\
      \pot{\tau_1 + \tau_2}{\inlex{i_1}}{\inrex{v_2}} &= 0 \\
      \pot{\tau_1 + \tau_2}{\inrex{i_2}}{\inlex{v_1}} &= 0 \\
      \pot{\tau_1 + \tau_2}{\inrex{i_1}}{\inrex{v_2}} &= \pot{\tau_2}{i_2}{v_2} \\
      \pot{\funty{\tau_1}{\tau_2}{\Theta}{\Theta_{\costfree}}}{\constidx}{\funex{f}{x}{e}} &= 1 \\
      \pot{\recty{\alpha}{\tau}}{\foldex{i}}{\foldex{v}} &= \pot{[\recty{\alpha}{\tau}/\alpha]\tau}{i}{v} \\
      &+
                                                           \textstyle\sum_{k \in \mkInd{\alpha}{\tau}{\foldex{i}}} \pot{[\recty{\alpha}{\tau}/\alpha]\tau}{k}{v} \\
      \pot{\recty{\alpha}{\tau}}{\lastidx}{\foldex{v}} &= 1
    \end{align*}
  \end{subfigure}
  \caption{Fundamental base polynomial index constructions.}
  \label{fig:indices}
\end{figure*}

Resource polynomials consist of a sum 
of ``monomial'' base polynomials with rational 
coefficients. We use indices to name those base 
polynomials. Figure~\ref{fig:indices} shows 
inference rules for the set of 
indices \(\mathcal{I}(\tau)\) for a given 
type \(\tau\). They nearly exactly mirror 
the syntactic values \(\mathcal{V}(\tau)\), with 
the addition of an ``\(\lastidx\)'' index for 
recursive types.\footnote{Several parts of the 
type system rely on describing
constant potential; we thus 
add \(\lastidx\) to do so for recursive types 
otherwise lacking such an index.} One possible 
intuition for an index is to view it like a 
pattern in a pattern match specifying a 
shape that values are compared against. 
However, matching a pattern
is a binary decision, whereas an 
index \emph{counts} occurrences in a value.

Following our running examples, 
both \(\mathsf{True}\) and \(\mathsf{False}\) are 
indices for \(\mathsf{bool}\) that match those values 
exactly; \(\mathsf{Nil}\) and \(\lastidx\) are 
indices for \(\mathsf{list}(\tau)\) that match 
against any list value exactly once; 
\(\mathsf{Cons}(\unitex, \mathsf{Nil})\) matches 
against any \(\mathsf{list}(\unitty)\) value
as many times as the length of the list; 
and \(\mathsf{Node}(\unitex, \mathsf{Nil})\) 
matches against any \(\mathsf{tree}(\unitty)\) value 
as many times as nodes in the tree.

\subsubsection*{Constant index set}
\label{sec:constindices}

A function is given in Figure~\ref{fig:indices} that defines a set of indices \(\mathcal{C}(\tau)\) for any type \(\tau\) such that the sum of their evaluation on any value of type \(\tau\) is exactly 1. The definition proceeds easily from the definition of index evaluation, which will be given shortly. (We also include a definition of \(\mathcal{C}\) for \(\topty\), a piece of syntax used in the course of the evaluation of \(\mathcal{M}\) that is substituted for the bound type variable when unfolding a recursive type, in order to only unfold each recursive type once.)

Following our running examples, we have \(\mathcal{C}(\mathsf{bool}) = \{\mathsf{True}, \mathsf{False}\}\), which indeed encompasses all possible values of type \(\mathsf{bool}\), and \(\mathcal{C}(\mathsf{list}(\tau)) = \mathcal{C}(\mathsf{tree}(\tau)) = \{\lastidx\}\), which forms the set of constant indices for any recursive type.

\subsubsection*{Recursive occurrence index set}
\label{sec:mkind}

This is the key insight that enables the extension to 
more general algebraic, mutually inductive types. The 
function \(\mkInd{\alpha}{\tau}{i}\) defined in 
Figure~\ref{fig:indices}, where \(\tau\) is a type with
no free type variables except for \(\alpha\) and \(i\) 
is an index for the type that \(\alpha\) represents, 
returns a set of indices that correspond to placing \(i\) 
at every occurrence of \(\alpha\) in \(\tau\). 
In more detail, here are the function's cases:
\begin{itemize}[left=\widthof{\(\tau_1 \to \tau_2\)}]
  \item[\(\alpha\).] We have found an occurence of \(\alpha\), so \(i\) goes here.
  \item[\(\topty\).] This represents some occurrence of a recursive type \emph{other} than the one \(\alpha\) refers to, having been substituted in during unfolding of said recursive type. Any occurrence of \(\alpha\) within that recursive type has already been handled by the \(\mathsf{fold}\) that is applied at the place of unfolding.
  \item[\(\unitty\).] No occurences of \(\alpha\) to be found here.
  \item[\(\sumty{\tau_1}{\tau_2}\).] No matter whether the value turns out to be a left or right injection, there could be a value of type \(\alpha\) within either, so we consider both cases.
  \item[\(\prodty{\tau_1}{\tau_2}\).] Here \(\alpha\) could occur inside \emph{both} projections of the pair, but we only want to consider one at a time, so we consider finding values in the first projection with arbitrary contents in the second, or vice versa.
  \item[\(\tau_1 \to \tau_2\).] We treat functions opaquely, with no \(\alpha\) values.
  \item[\(\recty{\beta}{\tau}\).] Here is the case critical for handling nested recursive types. As observed in \cref{sec:rosetrees}, introducing a \(\mathsf{fold}\) in the index here will cause this recursive process to happen over again during the \emph{evaluation} of the index, but for \(\beta\) instead of \(\alpha\). This sort of ``delaying'' of the recursive unrolling is what enables the nested recursive evaluation without having this process generate an infinite number of indices.
\end{itemize}

Following our running examples, we have \[\mkInd{\alpha}{\mathsf{bool}}{i} = \emptyset,\] because \(\alpha\) is not 
free in \(\mathsf{bool}\); \[\begin{array}{l} \mkInd{\alpha}{\sumty{\unitty}{\prodty{\mathsf{bool}}{\alpha}}}{i} = \\[0.2em] \hspace{2em}\{\inrex{(\pairex{\mathsf{True}}{i})}, \inrex{(\pairex{\mathsf{False}}{i})}\} \end{array}\] (where \(\sumty{\unitty}{\prodty{\mathsf{bool}}{\alpha}}\) is \(\mathsf{list}(\mathsf{bool})\) 
with the recursive binder stripped), because all 
recursive occurrences in a list of bools are at the 
tail of a cons cell with 
a bool
as the head; and \[\begin{array}{l} \mkInd{\beta}{\prodty{\mathsf{bool}}{\mathsf{list}(\beta)}}{i} = \\[0.2em] \hspace{2em} \{\pairex{\mathsf{False}}{\foldex{(\inrex{(\pairex{i}{\lastidx})})}}, \\ \hspace{2em}\phantom{\{} \pairex{\mathsf{True}}{\foldex{(\inrex{(\pairex{i}{\lastidx})})}}\} \end{array} \] (where \(\prodty{\mathsf{bool}}{\mathsf{list}(\beta)}\) is \(\mathsf{tree}(\mathsf{bool})\) with the recursive binder stripped), because all recursive occurrences in a rose tree of bools are in \emph{some} cons cell of the list of children. It's worth examining the last example a little more closely to grok the intuition for how this works for mutually inductive types: though the number of direct recursive occurrences of rose trees is unbounded and thus at first glance might require infinite indices to represent, the \(\foldex{}\) corresponding to the list \emph{itself} finds all of its recursive occurrences, allowing a finite number of indices to capture any number of descendants.

\subsubsection*{Index evaluation}
\label{sec:indexeval}

Finally, we reach the definition of the index evaluation function \(\pot{\tau}{i}{v}\) in Figure~\ref{fig:indices}, which evaluates the index \(i\) for type \(\tau\) on value \(v\). This gives the result of ``counting'' the number of matches of \(i\) in \(v\). The definition is straightforward except when evaluating an index \(\foldex{i}\), so we will just explain that rule in more detail. When evaluating index \(\foldex{i}\) on a value \(\foldex{v}\) of type \(\recty{\alpha}{\tau}\), we want to find all possible matches of \(i\) in \(v\). The first place those could occur is directly at the value \(v\), which the term \(\pot{[\recty{\alpha}{\tau}/\alpha]\tau}{i}{v}\) accounts for. However, we also want to consider matches in the recursive positions of the type within \(v\); as explained above, these positions are exactly what \(\mathcal{M}\{\alpha. \tau\}\) identifies, and we want to continue looking for all matches of \(\foldex{i}\) at those positions, so we sum the results of evaluating each index in \(\mkInd{\alpha}{\tau}{\foldex{i}}\) to count the recursive occurrences.

The results of evaluating indices on a few values of our example types are illustrated in Figure~\ref{fig:exindices}.

\newcommand{\treeone}{\begin{tikzpicture}[nodes={shape=circle,fill,inner sep=1pt}, node distance=1em] \node[draw=black] (A) {}; \node[below left=of A] (B) {}; \node[below right=of A] (C) {}; \draw (A) -- (B); \draw (A) -- (C); \end{tikzpicture}}
\newcommand{\treetwo}{\begin{tikzpicture}[nodes={shape=circle,fill,inner sep=1pt}, node distance=1em] \node[draw=black] (A) {}; \node[below left=of A] (B) {}; \node[below=of A] (C) {}; \node[below right=of A] (D) {}; \node[below=of B] (E) {}; \node[below=of D] (F) {}; \draw (A) -- (B); \draw (A) -- (C); \draw (A) -- (D); \draw (B) -- (E); \draw (D) -- (F); \end{tikzpicture}}

\begin{figure}
\renewcommand{\arraystretch}{1.1}
\centering
\begin{tabular}[t]{c|c|c|c}
  Type & Index & Value & Result \\
  \(\tau\) & \(i \in \mathcal{I}(\tau)\) & \(v \in \mathcal{V}(\tau)\) & \(\pot{\tau}{i}{v}\) \\\hline
  \(\mathsf{bool}\) & \(\mathsf{False}\) & \(\mathsf{False}\) & 1 \\
       & & \(\mathsf{True}\) & 0 \\[0.3em]
       & \(\mathsf{True}\) & \(\mathsf{False}\) & 0 \\
       & & \(\mathsf{True}\) & 1 \\[0.3em] \hline
  \(\mathsf{list}(\unitty)\) & \([]\) & \([\unitex; \unitex]\) & 1 \\
   & & \([\unitex; \unitex; \unitex; \unitex]\) & 1 \\[0.3em]
   & \([\unitex]\) & \([\unitex; \unitex]\) & 2 \\
   & & \([\unitex; \unitex; \unitex; \unitex]\) & 4 \\[0.3em]
   & \([\unitex; \unitex]\) & \([\unitex; \unitex]\) & 1 \\
       & & \([\unitex; \unitex; \unitex; \unitex]\) & 6 \\[0.3em] \hline
  \(\mathsf{tree}(\unitty)\) & \(\lastidx\) & \parbox{2.2em}{\treeone} & 1 \\
  & & \parbox{2.2em}{\treetwo} & 1 \\
  & \(\mathsf{Tree}(\unitex, [])\) & \parbox{2.2em}{\treeone} & 3 \\
  & & \parbox{2.2em}{\treetwo} & 6 \\
  & \(\mathsf{Tree}(\unitex, [\mathsf{Tree}(\unitex, [])])\) & \parbox{2.2em}{\treeone} & 2 \\
  & & \parbox{2.2em}{\treetwo} & 7 \\
\end{tabular}
  \caption{Example index evaluation results.}
  \label{fig:exindices}
\end{figure}

\vspace{1em}

With index evaluation defined, we may characterize the key property of the constant index set by induction on \(\tau\).
\begin{lemma}[Constant indices sum]
  For all types \(\tau\) and values \(v \in \mathcal{V}(\tau)\), \(\sum_{i \in \mathcal{C}(\tau)} \pot{\tau}{i}{v} = 1.\)
\end{lemma}

Note that distinct indices may sometimes refer to the same base polynomial. 
For example, the two indices \(\lastidx\) and \(\mathsf{Nil}\) for the type \(\mathsf{list}(\tau)\) 
both represent the constant function. 

\subsubsection*{Resource polynomials, proper}
\label{sec:respolysproper}

We have up to this point described the base polynomials by way of specification of their syntactic indices; the set of \emph{resource polynomials} \(\mathcal{R}(\tau)\) for a type \(\tau\) are then the linear combinations of base polynomials with nonnegative rational coefficients.

\subsection{Annotations}
\label{sec:annotations}

Though resource polynomials are the objects we really care about for analysis, the most useful representation of resource polynomials is a reified form we call \emph{annotations} \(\mathcal{A}(\tau)\).
\begin{definition}[Annotation]
  Let \(\mathcal{A}(\tau)\) denote the free \(\Qnn\)-semimodule with \(\mathcal{I}(\tau)\) as a basis. Then an \emph{annotation} for type \(\tau\) is an element of \(\mathcal{A}(\tau)\).
\end{definition}
Semimodules share the same definition as vector spaces, except for being defined over a semiring instead of a field. Because of this difference, in generality they lack much of the structure of vector spaces; however, \emph{free} semimodules can be thought to behave fairly similarly because they have bases. A longer overview of semimodules is available in Appendix~\ref{app:semimodules}.

We denote annotations as \(P\) or \(Q\) and
define \(p_i\) to be the coefficient corresponding to
index \(i\). Then we can recover the resource
polynomial as the potential
function \[ \Pot{\tau}{P}{v} \triangleq \textstyle\sum_{i \in \mathcal{I}(\tau)} p_i \cdot \pot{\tau}{i}{v}. \] Note that \(P \mapsto \Pot{\tau}{P}{\cdot}\) is a linear map from annotations to resource polynomials and that \(P \mapsto \Pot{\tau}{P}{v}\) is a linear form. In addition to basic operations on semimodules, we also use the following notions on annotations:
\begin{itemize}
  \item Set coercion to annotation for a finite set of indices \(B \subseteq \mathcal{I}(\tau)\), where \(b_i = 1\) if \(i \in B\) and 0 otherwise.
  \item Preorder \(P \leq Q\), defined by \(p_i \leq q_i\) for all \(i\). Note that \(P \leq Q\) is equivalent to the extension order, i.e. there exists an \(R\) such that \(Q = P + R\), and thus \(\Pot{\tau}{\cdot}{v}\) respects the order.
  \item Also note that because \(\mathcal{I}(\prodty{\tau_1}{\tau_2}) \cong \mathcal{I}(\tau_1) \times \mathcal{I}(\tau_2)\), we have \(\mathcal{A}(\prodty{\tau_1}{\tau_2}) \cong \mathcal{A}(\tau_1) \otimes \mathcal{A}(\tau_2)\). Then \(\Pot{\tau_1 \times \tau_2}{P \otimes Q}{\pairex{v_1}{v_2}} = \Pot{\tau_1}{P}{v_1} \cdot \Pot{\tau_2}{Q}{v_2}\). We will use the notations \(\mathsf{pair} : \mathcal{A}(\tau_1) \otimes \mathcal{A}(\tau_2) \to \mathcal{A}(\prodty{\tau_1}{\tau_2})\) and \(\mathsf{pair}^{-1} : \mathcal{A}(\prodty{\tau_1}{\tau_2}) \to \mathcal{A}(\tau_1) \otimes \mathcal{A}(\tau_2)\) to explicitly note the two sides of this isomorphism.
  \item Similarly, we have linear maps \(\mathsf{inl} : \mathcal{A}(\tau_1) \to \mathcal{A}(\sumty{\tau_1}{\tau_2})\) and its retraction \(\mathsf{inl}^{-1} : \mathcal{A}(\sumty{\tau_1}{\tau_2}) \to \mathcal{A}(\tau_1)\), and similarly for \(\mathsf{inr}\). Explicitly, \(\mathsf{inl}^{-1}\) is the linear map with defining equations \(\mathsf{inl}^{-1}(\inlex{i}) = i\) and \(\mathsf{inl^{-1}}(\inrex{i}) = 0\).
  \item The notation \(\text{id}_{\mathcal{A}(\tau)}\) refers to the identity linear map \(P \mapsto P : \mathcal{A}(\tau) \to \mathcal{A}(\tau)\).
\end{itemize}

Then, if we want, say, \(n^2\) potential for a unit list $\ell$ of length \(n\), we can use the annotation \(P = 1 \cdot [\unitex] + 2 \cdot [\unitex; \unitex]\), so that
\begin{align*}
  \Pot{\mathsf{list}(\unitty)}{P}{\ell} 
    &= 1 \!\cdot\! \pot{\mathsf{list}(\unitty)}{[\unitex]}{\ell}
     + 2 \!\cdot\! \pot{\mathsf{list}(\unitty)}{[\unitex; \unitex]}{\ell} \\
    &= n + 2 \tbinom{n}{2} = n^2.
\end{align*}

\subsubsection*{Shifting}
\label{sec:shifting}

A key requirement for our resource polynomials is the ability to fold and unfold recursive values while maintaining equal potential. Maintaining this ability was a key design constraint while constructing our system. We can accomplish this with the \emph{additive shift} operator.
\begin{definition}[Additive shift operator]
  Let \(\recty{\alpha}{\tau}\) be a recursive type. Then the \emph{additive shift operator}
 \(\shift\) is the linear map \(\shift : \mathcal{A}(\recty{\alpha}{\tau}) \to \mathcal{A}([\recty{\alpha}{\tau}/\alpha]\tau)\) corresponding to the function from basis elements \(\mathcal{I}(\recty{\alpha}{\tau})\) to \(\mathcal{A}([\recty{\alpha}{\tau}/\alpha]\tau)\) defined by
  \begin{align*}
    \shift\, \lastidx \defd \mathcal{C}([\recty{\alpha}{\tau}/\alpha]\tau) 
    & & 
    \shift (\foldex{i}) \defd i + \mkInd{\alpha}{\tau}{\foldex{i}}.
  \end{align*}
\end{definition}
In words, \(\lastidx\) is the constant index, and \(\foldex{i}\) refers to evaluation at both the current constructor (the \(i\) term) as well as all immediate children (the \(\mkInd{\alpha}{\tau}{\foldex{i}}\) term). The key property we desire for this operator is as follows:
\begin{theorem}[Shift preserves potential]
  For any \(P \in \mathcal{A}(\recty{\alpha}{\tau})\) and \(v \in \mathcal{V}([\recty{\alpha}{\tau}/\alpha]\tau)\), \[\Pot{\recty{\alpha}{\tau}}{P}{\foldex{v}} = \Pot{[\recty{\alpha}{\tau}/\alpha]\tau}{\shift P}{v}.\]
\end{theorem}
\begin{proof}
  This is equivalent to the statement of equality of linear forms \(\Pot{\recty{\alpha}{\tau}}{\cdot}{\foldex{v}} = \Pot{[\recty{\alpha}{\tau}/\alpha]\tau}{\shift \cdot}{v}\). By linearity, it suffices to show this on basis elements \(\mathcal{I}(\recty{\alpha}{\tau})\), at which point it follows directly.
\end{proof}
It can additionally be shown that shifting is in fact a linear isomorphism \(\mathcal{A}(\recty{\alpha}{\tau}) \cong \mathcal{A}([\recty{\alpha}{\tau}/\alpha]\tau)\).

\subsubsection*{Sharing}
\label{sec:sharing}

Since we need to be able to use a value multiple times, we need to be able to split its potential across multiple uses. Though this may sound simple at first, subtleties arise due to the multivariate setting: what if the potential between them ends up intertwined? For this we need the sharing operator, a bilinear map \(\annotshare : \mathcal{A}(\tau) \to \mathcal{A}(\tau) \to \mathcal{A}(\tau)\). Similarly to shifting, it suffices to define this just on basis elements. The full definition is available in Appendix~\ref{app:sharing}, but we consider it a key contribution of the paper, so we highlight the definition for the critical case, sharing two \(\mathsf{fold}\) indices:
\begin{align*}
   \annotshare_{\recty{\alpha}{\tau}}(\foldex{i}, \foldex{j}) &\defd \foldex{\left(\annotshare_{[\recty{\alpha}{\tau}/\alpha]\tau}(i, j)\right)} \\
   &+ \foldex{\left(\annotshare_{[\recty{\alpha}{\tau}/\alpha]\tau}(\mkInd{\alpha}{\tau}{\foldex{i}}, j)\right)} \\
   &+ \foldex{\left(\annotshare_{[\recty{\alpha}{\tau}/\alpha]\tau}(i, \mkInd{\alpha}{\tau}{\foldex{j}})\right)} \\
   &+ \foldex{\left(\mathcal{N}\{\alpha.\tau\}(\foldex{i}, \foldex{j})\right)}
\end{align*}
Here \(\mathcal{N}\) is like \(\mathcal{M}\), but places \emph{two} indices at \emph{two} occurrences of \(\alpha\). Intuitively, this says that a pair of indices can apply in the same value in any of these four categories of places: both at the current value, the left at a child and the right at the current value, the left at the current value and the right at a child, or both at a child.

The sharing operator satisfies the key property stated below:
\begin{theorem}[Share preserves potential]
  For \(P, Q \in \mathcal{A}(\tau)\) and \(v \in \mathcal{V}(\tau)\), \(\Pot{\tau}{P}{v} \cdot \Pot{\tau}{Q}{v} = \Pot{\tau}{P \annotshare Q}{v}\).
\end{theorem}

\subsection[Comparison to Hoffmann et al. {[2011]}]{Comparison to Hoffmann et al.~\cite{hoffmann2011multivariate}}
From this description of potential functions, it is not entirely clear whether our described resource polynomials are a \emph{generalization} of previous multivariate AARA potentials as in Hoffmann et al.~\cite{hoffmann2011multivariate}, or whether it is instead simply \emph{different}. We show it is the former with the following theorem:
\begin{theorem}
  All resource polynomials representable in Hoffmann et al.~\cite{hoffmann2011multivariate} are also representable in our system.
\end{theorem}
\begin{proof}[Proof sketch]
  Because resource polynomials are linear combinations of base polynomials, it suffices to show that their base polynomials are representable as our annotations. We show this by induction over types; the only nontrivial case is for binary trees. After a further induction over the length of the list of indices that serves as an index for such a binary tree, we can consider all possible splittings of the list, inductively obtain annotations for each splitting, and construct nodes with those annotations on either side.
\end{proof}
The proof is extensible to the finite arity trees as in \cite{hoffmann2017towards}.

\section{Language \& Type System}
\label{sec:typesystem}

Our language is essentially eager FPC~\cite{harper2016practical,fiore1994axiomatization} with a \(\tickex{q}\) expression to express cost and an explicit let construct. \(\mathsf{tick}\) expressions are the only sources of cost in our language, but any given cost metric based on syntactic forms can be desugared into a language with no cost other than explicit \(\mathsf{tick}\) expressions; additionally, such forms offer more flexibility for the programmer to specify particular kinds of costs.

\subsection{Semantics}
\label{sec:semantics}

The cost semantics of the language is a standard small-step operational semantics. The judgement \(\singlestep{e}{q}{e'}{q'}\) says the expression \(e\), starting with \(q \geq 0\) resources, transitions in a single step to expression \(e'\), with \(q' \geq 0\) resources remaining. The only reduction added compared to a pure call-by-value language is that for \(\tickex{q}\), as follows: \[ \inferrule{q \geq q_0}{\singlestep{\tickex{q_0}}{q}{\unitex}{q - q_0}} \] The judgement \(\multistep{e}{q}{e'}{q'}\) is then the transitive reflexive closure of the single step relation, with the constraint that only nonnegative resources may be considered. For use in cost-free derivations, we also define a ``pure'' semantics \(e \rightarrow e' \defd \exists q, q'.\, \singlestep{e}{q}{e'}{q'}\) and denote the transitive reflexive closure of that as \(e \rightarrow^* e'\).

\subsection{Type judgements}
\label{sec:typejudgements}

Type judgements in our system are as follows:

\begin{minipage}{\linewidth}
\vspace{2.5em}
\[
  \hasty{\Gamma}{\tikzmarknode{P}{\highlight{red}{$P$}}}{\tikzmarknode{c}{\highlight{OliveGreen}{$c$}}}{e}{\tau}{\tikzmarknode{Q}{\highlight{Fuchsia}{$Q$}}}
\]
\begin{tikzpicture}[overlay,remember picture,>=stealth,nodes={align=left,inner ysep=1pt},<-]
  \coordinate (Plabel) at ($(P.north)+(0,1em)$);
  \path (P.north) -- (Plabel) node[anchor=south east,color=red!67] (inannot){\textbf{input annotation on } $\Gamma$};
  \draw [color=red!87](P.north) |- ([xshift=-0.3ex,color=red]inannot.south west);

  \path (c.north) -- (c.north |- Plabel) node[anchor=south west,color=OliveGreen!67] (costmodel){\textbf{cost model}};
  \draw [color=OliveGreen!87](c.north) |- ([xshift=-0.3ex,color=OliveGreen]costmodel.south east);

  \path (Q.south) ++ (0,-1em) node[anchor=north east,color=Fuchsia!67] (remannot){\textbf{remainder annotation on } $\Gamma, \circ : \tau$};
  \draw [color=Fuchsia!57](Q.south) |- ([xshift=-0.3ex,color=Fuchsia]remannot.south west);
\end{tikzpicture}
\vspace{1.5em}
\end{minipage}

This can be informally read as 
``in the context \(\Gamma\) with \(P\) resources, 
under cost model \(c\), \(e\) has type \(\tau\) 
with \(Q\) resources left over.'' We now go into 
detail about each of these constructs that we 
have yet to explain: context annotations, 
remainder contexts, and cost models.

\subsubsection{Context annotations}
\label{sec:contextannotations}

Define \(\Idx(\Gamma) \defd \prod_{x : \tau \in \Gamma} \Idx(\tau)\) and \(\mathcal{A}(\Gamma)\) and \(\mathcal{C}(\Gamma)\) as the corresponding extensions of \(\mathcal{A}(\tau)\) and \(\mathcal{C}(\tau)\); note that \(\mathcal{A}(\Gamma) \cong \bigotimes_{x : \tau \in \Gamma} \mathcal{A}(\tau)\).
We may notationally elide coercions between isomorphic 
semimodules, such as \(\mathcal{A}(\Gamma_1, \Gamma_2) \cong \mathcal{A}(\Gamma_1) \otimes \mathcal{A}(\Gamma_2)\) 
and \(\mathcal{A}(x : \tau) \cong \mathcal{A}(\tau)\). 
Other annotation operations (all of which are linear maps) include:
\begin{itemize}
  \item The ``projection operator'' \(\proj{i}{\Gamma_1} : \mathcal{A}(\Gamma_1, \Gamma_2) \to \mathcal{A}(\Gamma_1)\), where \(i \in \mathcal{I}(\Gamma_2)\), which takes a slice of the annotation at the index \(i\). \(Q = \proj{i}{\Gamma_1}(P)\) is defined by \(q_j = p_{(j, i)}\).
  \item The ``pairing operator'' \(\mathsf{pair}^{x_1,x_2}_y : \mathcal{A}(\Gamma, x_1:\tau_1, x_2:\tau_2) \to \mathcal{A}(\Gamma, y:\prodty{\tau_1}{\tau_2})\) defined just by reassociating, because \(\mathcal{A}(\Gamma, x_1:\tau_1, x_2:\tau_2) \cong \mathcal{A}(\Gamma) \otimes \mathcal{A}(\prodty{\tau_1}{\tau_2}) \cong \mathcal{A}(\Gamma, y : \prodty{\tau_1}{\tau_2})\).
  \item An extension of the shift operator to context annotations, \(\shift_x = (\text{id}_{\mathcal{A}(\Gamma)} \otimes \shift) : \mathcal{A}(\Gamma, x : \recty{\alpha}{\tau}) \to \mathcal{A}(\Gamma, x : [\recty{\alpha}{\tau}/\alpha]\tau)\), i.e., shifting is applied to \(x\).
  \item An extension of the sharing operator to context annotations, \(\shareext{x}{y}{z} = (\text{id}_{\mathcal{A}(\Gamma)} \otimes \annotshare) : \mathcal{A}(\Gamma, x : \tau, y : \tau) \to \mathcal{A}(\Gamma, z : \tau)\), i.e., the variables \(x\) and \(y\) are shared together into \(z\)
\end{itemize}

\subsubsection{Remainder contexts}
\label{sec:remcontexts}

To more accurately track potential, we make use of annotations on \emph{remainder contexts}, which were inspired by IO-contexts from linear logic proof search~\cite{cervesato2000efficient,hodas1994logic} but introduced in the programmatic AARA setting by Kahn and Hoffmann~\cite{kahn2021automatic}. Remainder contexts contain both the typing context \(\Gamma\) with the additional pseudovariable ``\(\circ:\tau\)'' representing the result. Annotations on remainder contexts then represent the potential left on the whole context after the expression being typed has terminated.

Many of the benefits of remainder contexts noted 
in \cite{kahn2021automatic} extend to our setting. 
Such advantages include functions' ability to 
return potential back to their arguments after being 
called and the elimination of explicit sharing.

\subsubsection{Cost models}
\label{sec:costmodels}

We have two different cost models: ``cost-paid'', denoted \(\costpaid\), and ``cost-free'', denoted \(\costfree\). The former refers to an actual, cost-relevant execution (\(\mapsto^*\)), while the latter refers to a pure execution (\(\rightarrow^*\)). While we aren't concerned with pure executions directly, understanding them is necessary to transform mixed potential contexts across abstraction boundaries (in our case, function calls).

\subsection{Types}
\label{sec:types}

We explained of most of our language's types in \cref{sec:typesandvalues}, but discuss recursive types and function types in more detail here.

Support for more general recursive types are the core novel feature we add to AARA. As put forth in sections~\ref{sec:nestedpotential} and~\ref{sec:respolys}, we make substantial contributions in generalizing resource polynomials to them. We chose the word ``recursive'' to describe these types because that is indeed how they are constructed. However, we wish to stress that the distinction between recursive types and \emph{inductive} types is not very meaningful in our setting: for one, our functions have built-in general recursion, so recursive types do not grant any more expressive power there than inductive types would; furthermore, because we have only trivial (constant) resource polynomials over function values themselves, the landscape of resource polynomials is also unaffected by the distinction between recursive and inductive types.

Function types have the form \(\funty{\tau_1}{\tau_2}{\Theta}{\Theta_{\costfree}}\), where \(\Theta, \Theta_{\costfree} \subseteq \mathcal{A}(\tau_1) \times \mathcal{A}(\prodty{\tau_1}{\tau_2})\). \(\Theta\) and \(\Theta_{\costfree}\) are \emph{resource specifications}: they denote how much potential is required on a function's argument and how much potential is returned on the argument and result. A function may have many different resource specifications in order to handle resource polymorphic recursion. Intuitively, for any specification \((P, Q) \in \Theta\), to call the function on an argument \(v\), at least \(\Pot{\tau_1}{P}{v}\) resources are needed; once it returns with value \(v'\), \(\Pot{\prodty{\tau_1}{\tau_2}}{Q}{\pairex{v}{v'}}\) resources are returned as well. The argument \(v\) is mentioned in the output potential to enable the remainder contexts discussed in \cref{sec:remcontexts}. \(\Theta_{\costfree}\) serves a similar purpose for the cost-free model discussed in~\cref{sec:costmodels}: \((P, Q) \in \Theta_{\costfree}\) implies that \(\Pot{\prodty{\tau_1}{\tau_2}}{Q}{\pairex{v}{v'}} \geq \Pot{\tau_1}{P}{v}\), with no mention of resources gained or spent during execution.

\subsection{Typing rules}
\label{sec:typingrules}

\newcommand{\myinferrule}[3][]{\inferrule*[lab=#1]{#2}{#3}}

\newcommand{\subbox}[1]{#1}
\begin{figure*}
  \judgement{Resource typing}{\(\hasty{\Gamma}{P}{c}{e}{\tau}{Q}\)}
  \vspace{1em}

  \newlength{\boxwidth}
  \setlength{\boxwidth}{\textwidth-6.8pt} 


  \small
  \centering
    \begin{mathpar}
    \myinferrule[\textsc{T:TickPos}]
    {q \geq 0\!\!\!\!\! \\ P \otimes \mathcal{C}(\circ:\unitty) = Q {+} q \cdot \mathcal{C}(\Gamma, \circ:\unitty)}
    {\hasty{\Gamma}{P}{\costpaid}{\tickex{q}}{\unitty}{Q}}

    \inferrule*[lab=\textsc{T:TickNeg},leftskip=.38cm,rightskip=.38cm]
    {q \geq 0\!\!\!\!\! \\ P \otimes \mathcal{C}(\circ:\unitty) {+} q \cdot \mathcal{C}(\Gamma, \circ:\unitty) = Q}
    {\hasty{\Gamma}{P}{\costpaid}{\tickex{-q}}{\unitty}{Q}}

    \myinferrule[\textsc{T:InR}]
    { P = (\text{id}_{\mathcal{A}(\Gamma)} \otimes (\annotshare \circ (\text{id}_{\mathcal{A}(\tau_r)} \otimes \mathsf{inr}^{-1})))(Q) }
    {\hasty{\Gamma, x : \tau_r}{P}{c}{\inrex{x}}{\sumty{\tau_l}{\tau_r}}{Q}}
\\
    \myinferrule[\textsc{T:TickFree}]
    {P \otimes \mathcal{C}(\circ:\unitty) = Q}
    {\hasty{\Gamma}{P}{\costfree}{\tickex{q}}{\unitty}{Q}}

    \myinferrule[\textsc{T:Var}]
    {\share{\circ}{x}(Q) = P}
    {\hasty{\Gamma, x : \tau}{P}{c}{x}{\tau}{Q}}

    \myinferrule[\textsc{T:Unit}]
    {P \otimes \mathcal{C}(\circ:\unitty) = Q}
    {\hasty{\Gamma}{P}{c}{\unitex}{\unitty}{Q}}

    \myinferrule[\textsc{T:InL}]
    { P = (\text{id}_{\mathcal{A}(\Gamma)} \otimes (\annotshare \circ (\text{id}_{\mathcal{A}(\tau_l)} \otimes \mathsf{inl}^{-1})))(Q) }
    {\hasty{\Gamma, x : \tau_l}{P}{c}{\inlex{x}}{\sumty{\tau_l}{\tau_r}}{Q}}

    \myinferrule[\textsc{T:Let}]
    {\hasty{\Gamma}{P}{c}{e_1}{\tau'}{R} \\\\
     (\text{id}_{\mathcal{A}(\Gamma)} \otimes \text{id}_{\mathcal{A}(\tau')})(R) = (\text{id}_{\mathcal{A}(\Gamma)} \otimes \text{id}_{\mathcal{A}(\tau')})(S) \\\\
      \hasty{\Gamma, x : \tau'}{S}{c}{e_2}{\tau}{Q \otimes \mathcal{C}(x:\tau')}
    }
    {\hasty{\Gamma}{P}{c}{\letex{e_1}{x}{e_2}}{\tau}{Q}}

    \myinferrule[\textsc{T:CaseSum},narrower=0.5]
    { x' \text{ fresh} \\
      \share{x'}{x}(P') = P \\
      \share{x'}{x}(Q') = Q \\
      \forall n \in \{\mathsf{l},\mathsf{r}\}.\, \big( \\
      (\text{id}_{\mathcal{A}(\Gamma)} \otimes \mathsf{in}_{n}^{-1} \otimes \mathsf{in}_{n}^{-1})(P') = (\text{id}_{\mathcal{A}(\Gamma)} \otimes \mathsf{in}_{n}^{-1} \otimes \text{id}_{\mathcal{A}(\tau_n)})(R_n) \\
      \hasty{\Gamma, x : \sumty{\tau_l}{\tau_r}, x_n : \tau_n}{R_n}{c}{e_n}{\tau}{S_n} \\
      (\text{id}_{\mathcal{A}(\Gamma)} \otimes \mathsf{in}_{n}^{-1} \otimes \text{id}_{\mathcal{A}(\tau_n)} \otimes \text{id}_{\mathcal{A}(\tau)})(S_n) = (\text{id}_{\mathcal{A}(\Gamma)} \otimes \mathsf{in}_{n}^{-1} \otimes \mathsf{in}_{n}^{-1} \otimes \text{id}_{\mathcal{A}(\tau)})(R_n) \big)
    }
    {\hasty{\Gamma, x : \sumty{\tau_l}{\tau_r}}{P}{c}{\casesex{x}{x_l}{e_l}{x_r}{e_r}}{\tau}{Q}}

    \myinferrule[\textsc{T:Prod}]
    { 
      P = (\text{id}_{\mathcal{A}(\Gamma)} \otimes (\mathsf{pair}^{-1} \circ \annotshare \circ (\mathsf{pair} \otimes \text{id}_{\mathcal{A}(\circ : \prodty{\tau_1}{\tau_2})})))(Q)
    }
    {\hasty{\Gamma, x_1 : \tau_1, x_2 : \tau_2}{P}{c}{\pairex{x_1}{x_2}}{\prodty{\tau_1}{\tau_2}}{Q}}
    
    \myinferrule[\textsc{T:CaseProd},narrower=0.5]
    { P = (\text{id}_{\mathcal{A}(\Gamma)} \otimes (\annotshare \circ (\text{id}_{\mathcal{A}(\prodty{\tau_1}{\tau_2})} \otimes \mathsf{pair})))(R) \\
      \hasty{\Gamma, x : \prodty{\tau_1}{\tau_2}, x_1 : \tau_1, x_2 : \tau_2}{R}{c}{e}{\tau}{S} \\
      Q = (\text{id}_{\mathcal{A}(\Gamma)} \otimes (\annotshare \circ (\text{id}_{\mathcal{A}(\prodty{\tau_1}{\tau_2})} \otimes \mathsf{pair})) \otimes \text{id}_{\mathcal{A}(\tau)})(S)
    }
    {\hasty{\Gamma, x : \prodty{\tau_1}{\tau_2}}{P}{c}{\casepex{x}{x_1}{x_2}{e}}{\tau}{Q}}

    \myinferrule[\textsc{T:Fold}]
    {P = \share{\circ}{x}(\shift_{\circ}(Q))}
    {\hasty{\Gamma, x : [\recty{\alpha}{\tau} / \alpha]\tau}{P}{c}{\foldex{x}}{\recty{\alpha}{\tau}}{Q}}

    \myinferrule[\textsc{T:Unfold}]
    { P = \share{\circ}{x}(\shift^{-1}_{\circ}(Q)) }
    {\hasty{\Gamma, x : \recty{\alpha}{\tau}}{P}{c}{\unfoldex{x}}{[\recty{\alpha}{\tau} / \alpha]\tau}{Q}}
\\
    \myinferrule[\textsc{T:Fun}]
    {\forall (R, S) \in \Theta_{\costpaid}.\;
      \hasty{\Gamma, f : \funty{\tau_1}{\tau_2}{\Theta_{\costpaid}}{\Theta_{\costfree}}, x : \tau_1}
            {\constidxs{\Gamma} \otimes R}
            {\costpaid}
            {e}
            {\tau_2}
            {\constidxs{\Gamma} \otimes S} \\
     \forall (R, S) \in \Theta_{\costfree}.\;
      \hasty{\Gamma, f : \funty{\tau_1}{\tau_2}{\Theta_{\costpaid}}{\Theta_{\costfree}}, x : \tau_1} 
            {\constidxs{\Gamma} \otimes R}
            {\costfree}
            {e}
            {\tau_2}
            {\constidxs{\Gamma} \otimes S} \\
    }
    {\hasty{\Gamma}{P}{c}{\funex{f}{x}{e}}{\funty{\tau_1}{\tau_2}{\Theta_{\costpaid}}{\Theta_{\costfree}}}{Q}}

    \myinferrule[\textsc{T:Weaken}]
    {\hasty{\Gamma}{R}{c}{e}{\tau}{S} \\\\
     P = R \otimes \mathcal{C}(x:\sigma) \\
     Q = S \otimes \mathcal{C}(x:\sigma)
    }
    {\hasty{\Gamma, x:\sigma}{P}{c}{e}{\tau}{Q}}

    \myinferrule[\textsc{T:App}]
    { \forall i \in \constidxs{\Gamma}.\, (\proj{\Gamma \mapsto i}{\tau_1}(P), \proj{\Gamma \mapsto i}{\prodty{\tau_1}{\tau_2}}(Q)) \in \Theta_{\costpaid} \\\\
      \forall i \notin \constidxs{\Gamma}.\, (\proj{\Gamma \mapsto i}{\tau_1}(P), \proj{\Gamma \mapsto i}{\prodty{\tau_1}{\tau_2}}(Q)) \in \Theta_{\costfree}
    }
    {\hasty{\Gamma, f : \langle \tau_1 \to \tau_2, \Theta_{\costpaid}, \Theta_{\costfree} \rangle, x : \tau_1}{P}{c}{\appex{f}{x}}{\tau_2}{Q}}

    \myinferrule[\textsc{T:Relax}]
    {\hasty{\Gamma'}{R}{c}{e}{\tau'}{S} \\\\
      R \leq P \\
      Q \leq S
    }
    {\hasty{\Gamma}{P}{c}{e}{\tau}{Q}}

    \myinferrule[\textsc{T:Augment}]
    {\hasty{\Gamma}{R}{c}{e}{\tau}{S} \\\\
      P = R + T \\
      Q = S + T \otimes \mathcal{C}(\tau)
    }
    {\hasty{\Gamma}{P}{c}{e}{\tau}{Q}}
    \end{mathpar}

  \caption{Resource typing inference rules. See~\cref{sec:typingrules} for explanations of selected rules.}
  \label{fig:typingrules}
\end{figure*}

The typing rules for our language are available in Figure~\ref{fig:typingrules}. All rules are syntax-directed, except for the final three, which are structural. All rules apply to programs in \emph{let-normal form} to allow more precise accounting of potential. Of course, a preprocessing step may be added before typing to hide this restriction from the viewpoint of the user.

We now overview a representative subset of the \emph{interesting} typing rules given in Figure~\ref{fig:typingrules}.

\ActivateWarningFilters[scitundefined]
\subsubsection*{\textsc{T:Let}}
This rule is primarily interesting because of how simple it is compared to past work. In previous multivariate AARA works, this rule had to consider cost-free derivations of the term being bound, in order to handle mixed potential between the result and other variables in the context~\cite{hoffmann2011multivariate}. However, our remainder contexts remove this requirement because the remainder annotation may simply mention these mixed potentials. This improvement does not come for free; instead, cost-free derivations are needed for functions, but we believe that that is a more appropriate abstraction boundary.

\subsubsection*{\textsc{T:Pair}}
Some rules use tensor constructions that look complex but are really just victims of cumbersome bookkeeping; this rule is one of them. In words, the condition says that \(Q\) is just \(P\), but with \(x_1\) and \(x_2\) shared and then packed into \(\circ\).

\subsubsection*{\textsc{T:Fold} and \textsc{T:Unfold}}
These rules are the core of the potential annotation-based type system, taking advantage of the definition of the additive shift operator in~\cref{sec:shifting}.

\subsubsection*{\textsc{T:App}}
This function application rule essentially requires enough potential on the argument to use a cost-paid annotation of the function that is then returned into the remainder annotation, and then all of the mixed potential terms must be transformed according to cost-free annotations of the function.
\DeactivateWarningFilters[scitundefined]

\subsection{Typing example}
We now give an example of a novel typing that can be given with this type system. We present the code in a surface language with some light sugar that gives names to constructors of recursive types. Here the cost model is again the number of cons cells created.

Imagine implementing a mock filesystem. Each node is either a file, consisting of its name and its unstructured contents, or a directory, consisting of its name and a list of nodes contained within:
\begin{lstlisting}
type filesystem = File of string * string
                | Dir of string * filesystem list
\end{lstlisting}

The function \(\mathsf{attach}\), given a name 
and a filesystem, returns a list of pairs where the 
first element is the given name and the second 
is the name of a directory or file in 
the filesystem.
\begin{lstlisting}
let rec attach dname (acc, fs) = match fs with
  | File (fname, _) -> [(dname, fname)]
  | Dir (subdname, fss) ->
    (dname, subdname) :: foldl (attach dname) (acc, fss)
\end{lstlisting}
\newcommand{\stringtype}{\mathbb{S}}
(Here \(\mathsf{foldl} : (\mathsf{'b} \times \mathsf{'a} \to \mathsf{'b}) \to \mathsf{'b} \times \mathsf{list}(\mathsf{'a}) \to \mathsf{'b}\) is the left fold function on lists.) Without considering resource usage, this function can be given the type \(\mathsf{attach} : \stringtype \to \mathsf{list} (\stringtype \times \stringtype) \times \mathsf{filesystem} \to \mathsf{list} (\stringtype \times \stringtype)\) where \(\stringtype\) is an abbreviation for the string type. This function creates cons cells equal to the number of nodes in the filesystem. Accordingly, we can give \(\mathsf{attach}\) a typing that requires no potential on its first argument, \[1 \cdot ([], \mathsf{File}(\star, \star)) + 1 \cdot ([], \mathsf{Dir}(\star, [\lastidx]))\] potential on its second argument (where \(\star\) gives the constant potential on strings), and returns no potential on its output.

Building on \(\mathsf{attach}\), the second example function \(\mathsf{trans}\), 
given a representation of a filesystem, computes a list of every pair 
\((d, s)\) such that \(d\) is the name of an ancestor directory of 
the file or subdirectory named \(s\).\footnote{This is adapted from an example given by Hoffmann~\cite{hoffmannTypesPotentialPolynomial2011}.}
\begin{lstlisting}
let rec trans (acc, fs) = match fs with
  | File (_, _) -> []
  | Dir (dname, fss) ->
    foldl trans ((foldl (attach dname) (acc, fss)), fss)
\end{lstlisting}
Without considering resource usage, this function can be given the type \(\mathsf{trans} : \mathsf{list} (\stringtype \times \stringtype) \times \mathsf{filesystem} \to \mathsf{list} (\stringtype \times \stringtype)\). This creates cons cells quadratic in the number of nodes in the filesystem. Using our type system, we can say that \(\mathsf{trans}\) requires a potential of \[1 \cdot ([], \mathsf{Dir}(\star, [\mathsf{File}(\star, \star)])) + 1 \cdot ([], \mathsf{Dir}(\star, [\mathsf{Dir}(\star, [\lastidx])]))\] on its argument and returns no potential on its output. The twice-nested occurrences of the \(\mathsf{Dir}\) constructor in the argument's annotation indicates the quadratic potential usage.

\subsection{Soundness}
\label{sec:soundness}

Soundness here means that a well-typed term in a closed context is safe to execute
if starting with enough resources:
\begin{theorem}[Soundness]\label{thm:soundness}
  Assume \(\hasty{\cdot}{p}{\costpaid}{e}{\tau}{Q}\). For any \(r \geq 0\) and execution \(\multistep{e}{p + r}{e'}{q}\), we have either:
  \begin{itemize}
    \item \(e' \in \mathcal{V}(\tau)\) and \(q \geq r + \Pot{\tau}{Q}{e'}\), or
    \item there exists some \(e''\) and \(q'\) such that \(\singlestep{e'}{q}{e''}{q'}\).
  \end{itemize}
\end{theorem}
\begin{proof}
  We construct a step-indexed logical relation; the only non-standard case is that for the arrow type, in which, starting at the next step, all of the type's annotations must be satisfied in all states with sufficiently many resources. The full definition is available in Appendix~\ref{app:soundness}. We then prove the fundamental theorem of logical relations by induction on the typing derivation and show that inhabitation in the logical relation implies the conclusion of this theorem.
\end{proof}

\section{Related Work}
\label{sec:relatedwork}

Our work builds upon the literature of the 
Automatic Amortized Resource Analysis (AARA) type 
system. AARA was first introduced in \cite{hofmann2003static},
using potential method~\cite{tarjan1985amortized} reasoning for
the automatic derivation of heap-space 
bounds \textit{linear} in the sizes of data structures.
AARA has been extended to support bounds in the forms
of polynomials \cite{hoffmann2010amortized}, exponentials
\cite{kahn2020exponential}, logarithms \cite{hofmann2021type},
and maxima \cite{kahn2021automatic, campbell2009amortised}. 
However, each of these works bakes their size parameters
and resource functions into the inductive data types of trees and lists 
(or labeled trees \cite{hoffmann2017towards}).
The one exception is the Schopenhauer language \cite{jost2010static}, which included support for deriving bounds on programs using nested recursive types, but only in the very restricted class of linear functions.
Our indices provide the first method of automatically constructing
resources functions and size parameters for general recursive types.
Our resulting system conservatively extends the bounds given by
the \textit{multivariate} polynomial system \cite{hoffmann2011multivariate,hoffmann2017towards},
wherein bounds may depend on the products of the sizes of 
data structures.  
And, while our work does not go in these directions, AARA's analysis can also cover other models of 
computation and analysis,
including imperative \cite{carbonneaux2015compositional}, 
object-oriented \cite{hofmann2006type, hofmann2013automatic}, 
probabilistic \cite{ngo2018bounded, wang2020raising}, and
parallel computation \cite{hoffmann2015automatic}, 
digital contract protocols \cite{das2021resource},
and lower bound costs \cite{dehesa2017verifying}.
Combinatorial 
species~\cite{bergeronCombinatorialSpeciesTreelike1997} 
may also relate to AARA's resource polynomials.

Aside from AARA, other type-based systems have also been 
used to analyze the resource usage of programs. These include 
linear dependent types \cite{dal2013geometry, dal2011linear},
refinement types \cite{radivcek2017monadic, wang2017timl},
modal types \cite{rajani2021unifying},
size types \cite{vasconcelos2008space, serrano2014resource},
annotation-based systems \cite{crary2000resource, danielsson2008lightweight},
and more. These type systems bookkeep costs 
using a variety of differing ideas, but they all 
enjoy the high composability provided by type systems, usually 
employ some linear features and cost constraints like AARA. 
Unlike AARA, however, many trade some degree of 
automatabilitiy for richer features -- at the extreme 
other end from AARA one finds type-based proof logics like in \cite{niu2022cost}, which require 
significant user work to prove cost bounds.

The term-rewriting space provides some work that is 
comparable with ours. Specifically, the system from \cite{hofmann2015multivariate} 
generalizes multivariate potential functions over aribtrary 
types using tree automata. We speculate that this system
is general enough to contain the resource functions 
generated by our approach.
However, their work leaves the open questions of how to pick the appropriate automata, and 
how to solve the constraints they induce. There has been 
much work \cite{bauer2019decidability} to even solve simpler cases 
than the multivariate case.
Our work could potentially be seen as a step toward this automation.
Other resource analysis work using term-rewriting 
include \cite{avanzini2016combination, avanzini2015analysing, hirokawa2014automated, naaf2017complexity}.

Recurrence relations are another common approach to 
resource analysis, especially in a functional setting 
\cite{kincaid2017compositional, danner2015denotational, kavvos2019recurrence}. 
Some recent work uses potential-based reasoning for amortization 
\cite{cutler2020denotational}. Usually these methods operate 
by extracting recurrences from the code and then solving them.
While this can be more difficult than solving the linear 
constraints extracted by AARA, it can allow the
expression of bounding functions that AARA cannot yet support. 

Techniques from imperative \cite{gulwani2009speed, gulwani2009speed2} 
and logic \cite{navas2007user} program analyses also can  
reason about cost in terms of general notions of data structure size.
However, that work does so with manually-defined notions of size 
which are reduced to numerical analysis. While such numerical analysis 
is common in cost analysis, it does not focus on the
the sorts of intrinsic features of data structures that our work does.

Other approaches to cost analysis include
abstract interpretation \cite{albert2007costa, albert2015non, lopez2018interval},
loop analyses \cite{blanc2010abc, kincaid2017non},
relational cost analysis \cite{cciccek2017relational, qu2021relational},
ranking functions \cite{chatterjee2019non}, 
and program logics \cite{gueneau2018fistful, atkey2010amortised, mevel2019time}. 
The field varies broadly and mixes many approaches. In particular,
the cited program logics make use of potential-based reasoning 
like AARA.




\section{Conclusion and Future Work}
\label{sec:conclusion}

This work's contributions include the extension of multivariate 
resource polynomials to regular recursive datatypes and the 
introduction of semimodules as a helpful formalism
in the system's specification. The extended resource 
polynomials enable the resource analysis of 
programs using complex, nested data structures like rose trees.
These comprise a major step forward in automatable resource 
analysis through the structural combinatorics of 
data types. 
Future work will include the implementation of 
these resource functions 
in a fully automated resource analysis typechecker,
like RaML \cite{hoffmann2017towards}.
Further, we would 
like to extend our methods 
for generating resouce polynomials
to also cover resource 
exponentials \cite{kahn2020exponential},
which currently are not supported 
at the multivariate level.


\section*{Acknowledgments}
This article is based on research supported by
the Algorand Centres of
Excellence programme managed by the Algorand Foundation
and by the National Science Foundation under awards 1801369, 1845514,
and 2007784.
Any opinions, findings, and conclusions contained in this document are
those of the authors and do not necessarily reflect the views of the
sponsoring organizations.

\bibliographystyle{IEEEtran}
\bibliography{IEEEabrv,sources}

\appendices

\clearpage

\section{Semimodules Overview}
\label{app:semimodules}

This section gives a brief overview of \emph{semimodules} and the operations on and properties of them that are relevant in this work.

An \(R\)-semimodule \(M\) is a generalization of a vector space that works over a semiring \(R\) instead of a field. In this work, we will assume that \(R\) is commutative, in order to avoid repeating ``commutative'' everywhere. Just as in vector spaces, there is an addition operation \(+ : M \times M \to M\) that forms a commutative monoid and a scalar multiplication operation \(\cdot : R \times M \to M\) that respects the various distributive and identity laws. The simplest example of an \(R\)-semimodule is \(R\) itself, with semimodule operations inherited from the ring structure.

A linear map (known also as a homomorphism) between \(R\)-semimodules \(M\) and \(M'\) is a function from \(M\) to \(M'\) that respects addition and scalar multiplication. An isomorphism between semimodules is a bijective linear map between them. A bilinear map \(M \times M' \to M''\) is a function that is linear in both arguments.

A free \(R\)-semimodule is one that is finitely generated by a basis \(E\) over the semimodule operations. Equivalently, the free \(R\)-semimodule generated by \(E\) is the set of finite maps from \(E\) to nonzero elements of \(R\), with addition and scalar multiplication defined pointwise. (Unlike vector spaces, not every semimodule is isomorphic to a free semimodule, but that is not relevant for this paper.)

If \(f, g : M \to M'\) are linear maps where \(M\) is the free \(R\)-semimodule generated by \(E\), then \(f = g\) iff \(f(e) = g(e)\) for all \(e \in E\), as can be derived using the properties of linear maps by inducting over an \(M\)-element's finite generation.

The tensor product of two \(R\)-semimodules \(M_1\) and \(M_2\) is the semimodule \(M_1 \otimes M_2\) with canonical linear map \(\mathbin{\otimes} : M_1 \times M_2 \mapsto m_1 \otimes m_2\) satisfying the universal property that all bilinear maps \(f : M_1 \times M_2 \to M'\) have a unique linear map \(\tilde{f} : M_1 \otimes M_2 \to M'\) such that \(f = \tilde{f} \circ \otimes\). (Its existence can be shown with a quotient construction.) The tensor product has symmetric monoidal structure, so \(M_1 \otimes M_2 \cong M_2 \otimes M_1\) and \(M_1 \otimes (M_2 \otimes M_3) \cong (M_1 \otimes M_2) \otimes M_3\). A tensoring operation \(f \otimes g\) on functions \(f : M_1 \to M_1'\) and \(g : M_2 \to M_2'\) can be defined such that \((f \otimes g)(m_1 \otimes m_2) = f(m_1) \otimes g(m_2)\) for any \(m_1 \in M_1\) and \(m_2 \in M_2\). By convention, for a function \(f : M_1 \to M_1'\), the notation \(f \otimes M_2\) denotes the map \(f \otimes g : M_1 \otimes M_2 \to M_1' \otimes M_2\) where \(g\) is the identity map on \(M_2\).


\section{Sharing Operator}
\label{app:sharing}

\newcommand{\m}[3]{\mkInd{#1}{#2}{#3}}
\newcommand{\n}[5]{\mathcal{N}\{{#1}[{#2}].{#3}\}({#4},{#5})}
\newcommand{\mergeop}[3]{#2 \curlyvee_{#1} #3}

\begin{figure*}

\begin{align*}
    \n \alpha \tau \sigma {i_1} {i_2} &\in \mathcal{A}([\recty{\alpha}{\tau}/\alpha] \sigma) \\[0.5em]
    \n \alpha \tau \unitty i j &= 0 \\[0.5em]
    \n \alpha \tau \topty i j &= 0 \\[0.5em]
    \n \alpha \tau {\sumty \sigma \rho} i j &= 
    \inlex {(\n \alpha \tau \sigma i j)} +
    \inrex {(\n \alpha \tau \rho i j)} \\[0.5em]
    \n \alpha \tau {\prodty \sigma \rho} i j &=
    \pairex {\n \alpha \tau \sigma i j} {\constidxs \rho} + \pairex {\constidxs \sigma} {\n \alpha \tau \rho i j} \\
    &+ \pairex {\m \alpha \sigma i} {\m \alpha \rho j} + \pairex {\m \alpha \sigma j} {\m \alpha \rho i}  \\[0.5em]
    \n \alpha \tau {\recty \beta \sigma} i j &= 
    \foldex {(\n \alpha \tau {[\topty/\beta]\sigma} i j)} \\
    &+ \foldex {(\mergeop {[\recty \tau \alpha/\alpha][\recty \beta \sigma/\beta]\sigma} {\m \beta {[\recty \alpha \tau/\alpha]\sigma} {\m\alpha {[\topty/\beta]\sigma} i}} {\m\alpha {[\topty/\beta]\sigma} j})} \\
    &+ \foldex {(\mergeop {[\recty \tau \alpha/\alpha][\recty \beta \sigma/\beta]\sigma} {\m\alpha {[\topty/\beta]\sigma} i} {\m \beta {[\recty \alpha \tau/\alpha]\sigma} {\m\alpha {[\topty/\beta]\sigma} j}})} \\
    &+ \foldex {(\n \beta \sigma {[\recty \alpha \tau/\alpha]\sigma} {\m\alpha {[\topty/\beta]\sigma} i} {\m\alpha {[\topty/\beta]\sigma} j})}
\end{align*}

\begin{align*}
    \mergeop \tau {i_1} {i_2} &\in \mathcal{A}(\tau) \\[0.5em]
    \mergeop \unitty \unitex \unitex &= \unitex \\[0.5em]
    \mergeop {\sumty \tau \sigma}  {(\inlex i)} {(\inlex j)} &= \inlex {\mergeop \tau i j } \\[0.5em]
    \mergeop {\sumty \tau \sigma} {(\inrex i)} {(\inrex j)} &= \inrex {\mergeop \sigma i j} \\[0.5em]
    \mergeop {\prodty \tau \sigma} {\pairex {i_1} {i_2}} {\pairex {j_1} {j_2}} 
    &=\pairex {\mergeop \tau {i_1} {j_1}} {\mergeop \sigma {i_2} {j_2}} \\[0.5em]
    \mergeop {\recty \alpha \tau} \lastidx \lastidx &= \lastidx \\[0.5em]
    \mergeop \topty \lastidx \lastidx &= \lastidx \\[0.5em]
    \mergeop {\recty \alpha \tau} {(\foldex i)} \lastidx &= \foldex i \\[0.5em]
    \mergeop {\recty \alpha \tau} \lastidx {(\foldex i)} &= \foldex i \\[0.5em]
    \mergeop {\recty \alpha \tau} {(\foldex i)} {(\foldex j)} &=
    \foldex {(\mergeop {[\recty \alpha \tau / \alpha] \tau} i j)} \\
    &+ \foldex {(\mergeop {[\recty \alpha \tau / \alpha] \tau} {\m \alpha \tau i} j)} \\
    &+ \foldex {(\mergeop {[\recty \alpha \tau / \alpha] \tau} i {\m \alpha \tau j})} \\
    &+ \foldex {(\n \alpha \tau \tau i j)} \\[0.5em]
    \mergeop \tau i j &= 0 \hspace{3em} \mathrm{otherwise}
\end{align*}

  \caption{Definition of sharing operator}
  \label{fig:sharing}
\end{figure*}

See Figure~\ref{fig:sharing}.

\section{Soundness Details}
\label{app:soundness}

\newcommand{\interp}[1]{\llbracket #1 \rrbracket}
\newcommand{\vallr}[1]{\mathcal{V} \interp{#1}_n}
\newcommand{\vallrm}[1]{\mathcal{V} \interp{#1}_m}
\newcommand{\vallrmp}[1]{\mathcal{V} \interp{#1}_{m'}}
\newcommand{\vallrp}[1]{\mathcal{V'} \interp{#1}_{n+1}}
\newcommand{\vallrnp}[1]{\mathcal{V'} \interp{#1}_n}
\newcommand{\exlr}[2]{\mathcal{E}_c \interp{\langle #1, #2 \rangle}}
\newcommand{\exlrcf}[2]{\mathcal{E}_{\costfree} \interp{\langle #1, #2 \rangle}}
\newcommand{\wpr}[2]{\mathcal{WP}_{#1} [#2]}

\begin{figure*}
  \judgement{Weakest Precondition Relation}{\(\wpr{c}{\Psi}_n(e, q)\)}
  \begin{align*}
    \wpr{\costpaid}{\Psi}_0(e, q) &\defd \top \\[0.5em]
    \wpr{\costpaid}{\Psi}_{n+1}(e, q)
    &\defd (e \text{ val} \land \Psi(n+1, e, q)) \lor \strut \\
    &\hspace{1.3em} ((\exists e', q'.\; \singlestep{e}{q}{e'}{q'}) \land \forall e', q'.\; \singlestep{e}{q}{e'}{q'} \Rightarrow \wpr{\costpaid}{\Psi}_n(e', q')) \\[0.5em]
    \wpr{\costfree}{\Psi}_{n+1}(e, q)
    &\defd (e \text{ val} \land \Psi(n+1, e, q)) \lor \strut \\
    &\hspace{1.3em} ((\exists e'.\; e \rightarrow e') \land \forall e'.\; e \rightarrow e' \Rightarrow \wpr{\costfree}{\Psi}_n(e', q))
  \end{align*}

  \judgement{Value Logical Relation}{\(\vallr{\tau}(v)\)}
  \begin{align*}
    \mathcal{V'} \interp{\tau}_0(v) &\defd \top \\[1em]
    \vallrp{\unitty}(v) &\defd v = \unitex \\
    \vallrp{\prodty{\tau_1}{\tau_2}}(v) &\defd \exists v_1, v_2.\; v = \pairex{v_1}{v_2} \land \vallrp{\tau_1}(v_1) \land \vallrp{\tau_2}(v_2) \\
    \vallrp{\sumty{\tau_1}{\tau_2}}(v) &\defd (\exists v_1.\; v = \inlex{v_1} \land \vallrp{\tau_1}(v_1)) \lor (\exists v_2.\; v = \inrex{v_2} \land \vallrp{\tau_2}(v_2)) \\
    \vallrp{\funty{\tau_1}{\tau_2}{\Theta_{\costpaid}}{\Theta_{\costfree}}}(v) &\defd \exists f, x, e.\; v = \funex{f}{x}{e} \land \forall m \leq n. \\
    &\hspace{2em} (\forall (P, Q) \in \Theta_{\costpaid}.\; \forall q \geq 0, v'.\; \vallrm{\tau_1}(v') \Rightarrow \\
    &\hspace{4em} \wpr{\costpaid}{\lambda (m', v'', q').\; \vallrmp{\tau_2}(v'') \land\strut \\
    &\hspace{8em} q' \geq q + \Pot{\prodty{\tau_1}{\tau_2}}{Q}{\pairex{v'}{v''}}}_m \\
    &\hspace{6.4em} ([v'/x, v/f]e, q + \Pot{\tau_1}{P}{v'}) \land\strut \\
    &\hspace{2em} (\forall (P, Q) \in \Theta_{\costfree}.\; \forall q \geq 0, v'.\; \vallrm{\tau_1}(v') \Rightarrow \\
    &\hspace{4em} \wpr{\costfree}{\lambda (m', v'', q').\; \vallrmp{\tau_2}(v'') \land\strut \\
    &\hspace{8em} q' \geq q + \Pot{\prodty{\tau_1}{\tau_2}}{Q}{\pairex{v'}{v''}}}_m \\
    &\hspace{6.4em} ([v'/x, v/f]e, q + \Pot{\tau_1}{P}{v'}) \\
    \vallrp{\recty{\alpha}{\tau}}(v) &\defd \exists v'.\; v = \foldex{v'} \land \vallrnp{[\recty{\alpha}{\tau}/\alpha]\tau}(v') \\[0.5em]
    \mathcal{V} \interp{\tau}_n (v) &\defd v \in \mathcal{V}(\tau) \land \mathcal{V'} \interp{\tau}_n (v)
  \end{align*}

  \caption{Weakest precondition relation and logical relation on values.}
  \label{fig:logrels}
\end{figure*}

To state our notion of semantic well-typedness, in Figure~\ref{fig:logrels} we first define a weakest precondition unary relation \(\mathcal{WP}\) on expressions and a type-indexed unary logical relation \(\mathcal{V}\) on values. Both are indexed by a cost model \(c\), as explained in \cref{sec:costmodels}.

Our weakest precondition relation \(\mathcal{WP}\) is parameterized by a postcondition \(\Psi\) and takes as arguments an expression \(e\) and initial resources \(q\). Intuitively, \(\wpr{c}{\Psi}_n(e, q)\) holds if \(e\) ``can't go wrong'' in fewer than \(n\) steps starting with \(q\) resources and \(\Psi(m, v, q')\) holds on any resulting value \(v\), remaining steps \(m\), and remaining resources \(q'\).

The core theorem for relating cost-paid and cost-free weakest preconditions, as needed by the application rule, is that we can combine both into one cost-paid weakest precondition; formally:
\begin{lemma}[Weakest precondition glomming]\label{lem:wpcombine} 
  If \(\wpr{\costpaid}{\Psi_p}(e, q_p)\) and \(\wpr{\costfree}{\Psi_f}(e, q_f)\) both hold, then so does \(\wpr{\costpaid}{\lambda (v, q').\; \exists q_p', q_f'.\, q' = q_p' + q_f' \land \Psi_p(v, q_p') \land \Psi_f(v, q_f')}(e, q_p + q_f)\).
\end{lemma}
\begin{proof} By induction over steps and use of the definition of \(\rightarrow\). \end{proof}

The value relation is completely standard for all types but function types, so we will only detail arrows. As explained in \cref{sec:types}, an arrow type has form \(\funty{\tau_1}{\tau_2}{\Theta}{\Theta_{\costfree}}\) where \(\Theta, \Theta_{\costfree} \subseteq \mathcal{A}(\tau_1) \times \mathcal{A}(\prodty{\tau_1}{\tau_2})\), which are sets of specifications of the potential needed and returned by the function. Our interpretation thus intuitively states that for each annotation specification \((P, Q) \in \Theta\) (respectively \(\Theta_{\costfree}\)), for all resource states \(q\) and arguments \(v'\), assuming that \(v'\) is semantically well typed at \(\tau_1\), running the function with \(q\) resources augmented with the added potential that \(P\) mandates of \(v\), the result (if terminating) is semantically well typed at \(\tau_2\) and the remaining resources are at least \(q\) augmented with that which \(Q\) guarantees. Here the quantification over all smaller step-indices \(m\) and resources \(q\) acts as a Kripke-style quantification over future states.

\newcommand{\hassemty}[6]{{#1} ; {#2} \vDash_{#3} {#4} : {#5} ; {#6}}

With those relations under our belt, we define the semantic resource typing judgement:
\begin{definition}[Semantic resource typing]
\[\begin{array}{l}
  \hassemty{\Gamma}{P}{c}{e}{\tau}{Q} \defd \\
  \hspace{1.3em} \forall \gamma, n, q \geq 0.\; \mathcal{V}\interp{\Gamma}_n(\gamma) \Rightarrow \\
  \hspace{2em} \wpr{c}{\lambda (m, v, q').\; \mathcal{V}\interp{\tau}_m(v) \land \\
  \hspace{8.3em} q' \geq q + \Pot{\Gamma, \circ:\tau}{Q}{\gamma[\circ\mapsto v]}} \\
  \hspace{4.3em} (e, q + \Pot{\Gamma}{P}{\gamma})
\end{array}\]
That is, under cost model \(c\) and in annotated context \(\langle \Gamma; P \rangle\), expression \(e\) is semantically well-typed at \(\tau\) with remainder annotation \(Q\) iff for all semantically well-typed closing substitution contexts \(\gamma\), \(e\) never crashes and results in a semantically well-typed remainder value context with potential at least as much as according to \(Q\) if it terminates, assuming the run is started with at least as many resources as \(P\) says \(\gamma\) should.
\end{definition}

At this point, we wish to prove the fundamental theorem of logical relations:
\begin{lemma}[Fundamental theorem]
  If \(\hasty{\Gamma}{P}{c}{e}{\tau}{Q}\), then \(\hassemty{\Gamma}{P}{c}{e}{\tau}{Q}\).
\end{lemma}
\begin{proof}[Proof sketch]
 Given the way we've set things up, this is actually quite trivial to prove for almost all rules, as the statement boils down to equality of potential evaluations that are easily discharged with theorems from \cref{sec:respolys} and \cref{sec:contextannotations}. For those rules where this is not true, we give sketches of proofs:
 \begin{itemize}
  \item \textsf{T:Let}. Inducts over the execution derivation for \(e_1\) and uses the evaluation context stepping rule for \(\mathsf{let}\).
  \item \textsf{T:Unfold}. The equality of potentials follows immediately; the only typing wrinkle comes from the ``later'' in the \(\mu\) type, but this is eliminated with the \(\unfoldex{(\foldex{v})}\) to \(v\) step.
  \item \textsf{T:App}. First picks the one cost-paid annotation in \(\Theta\) to be used, then inducts over the finite set of nonconstant indices in \(P\) with nonzero coefficients to create one big weakest precondition for the function application, using Lemma~\ref{lem:wpcombine}.
  \item \textsf{T:Relax}. Follows from annotation evaluation respecting the ordering on annotations.
 \end{itemize}
\end{proof}

Now we can state our adequacy result:
\begin{lemma}[Adequacy]
  Assume \(\hassemty{\cdot}{p}{\costpaid}{e}{\tau}{Q}\). For any \(r \geq 0\) and execution \(\multistep{e}{p + r}{e'}{q}\), we have either:
  \begin{itemize}
    \item \(e' \in \mathcal{V}(\tau)\) and \(q \geq r + \Pot{\tau}{Q}{e'}\), or
    \item there exists some \(e''\) and \(q'\) such that \(\singlestep{e'}{q}{e''}{q'}\).
  \end{itemize}
\end{lemma}

Theorem~\ref{thm:soundness} is a corollary of the above two lemmas.

\end{document}

%% file: macros.tex
\newcommand{\highlight}[2]{\colorbox{#1!17}{#2}}
\newcommand{\highlightdark}[2]{\colorbox{#1!47}{#2}}

\newcommand{\lap}{\mathrm{Lap}}
\newcommand{\pr}{\mathrm{Pr}}

\newcommand{\Tset}{\mathcal{T}}
\newcommand{\Dset}{\mathcal{D}}
\newcommand{\Rbound}{\widetilde{\mathcal{R}}}

\newcommand{\scalemath}[2]{\scalebox{#1}{\mbox{\ensuremath{\displaystyle #2}}}}

\makeatletter
\lstdefinelanguage{aara}{%
  basicstyle={},
  morekeywords=[1]{let,in,rec,match,with,type,of,if,then,else},
  morekeywords=[2]{},
  morekeywords=[3]{},
  morekeywords=[4]{},
  columns=[l]fullflexible,
  texcl=true,
  mathescape=true,
  identifierstyle={\sffamily},
  keywordstyle=[1]{\sffamily\color{blue}},
  keywordstyle=[2]{\sffamily\color{blue}},
  keywordstyle=[3]{\color{red}},
  keywordstyle=[4]{\rmfamily\itshape},
  showspaces=false,
  literate={->}{$\rightarrow\,$}{1}
           {|}{|\hspace{0.4em}}{1}
           {...}{\ldots}{1},
  breaklines=false %
}
\makeatother

\newcommand{\bnfdef}{\Coloneqq}
\newcommand{\bnfalt}{\;\mid\;\;}

\newcommand{\recty}[2]{\mu {#1} .\, {#2}}
\newcommand{\unitty}{\mathbf{1}}
\newcommand{\prodty}[2]{#1 \times #2}
\newcommand{\sumty}[2]{#1 + #2}
\newcommand{\funty}[4]{\langle {#1} \to {#2}, {#3}, {#4} \rangle}
\newcommand{\topty}{\top_{\mu}}

\newcommand{\unitex}{\mathsf{tt}}
\newcommand{\pairex}[2]{\mathsf{pair}({#1};\, {#2})}
\newcommand{\inlex}[1]{\mathsf{inl}\, {#1}}
\newcommand{\inrex}[1]{\mathsf{inr}\, {#1}}
\newcommand{\inex}[2]{\mathsf{in}_{#1}\, {#2}}
\newcommand{\closex}[4]{\mathsf{C}({#1}, {#2},{#3}.\, {#4})}
\newcommand{\foldex}[1]{\mathsf{fold}\, {#1}}
\newcommand{\unfoldex}[1]{\mathsf{unfold}\, {#1}}
\newcommand{\funex}[3]{\mathsf{fun}({#1}, {#2}.\, {#3})}
\newcommand{\tickex}[1]{\mathsf{tick}\{{#1}\}}
\newcommand{\letex}[3]{\mathsf{let}({#1};\, {#2}.\, {#3})}
\newcommand{\casesex}[5]{\mathsf{cases}({#1};\, {#2}.\, {#3};\, {#4}.\, {#5})}
\newcommand{\casepex}[4]{\mathsf{casep}({#1};\, {#2}, {#3}.\, {#4})}
\newcommand{\appex}[2]{\mathsf{app}({#1};\, {#2})}

\newcommand{\lastidx}{\mathsf{end}}
\newcommand{\constidx}{\lambda}
\newcommand{\lamidx}{\mathsf{lam}}
\newcommand{\inidx}[1]{\mathsf{in}\, {#1}}
\newcommand{\remidx}[1]{\mathsf{rem}\, {#1}}
\newcommand{\outidx}[1]{\mathsf{out}\, {#1}}

\newcommand{\constidxs}[1]{\mathcal{C} ({#1})}
\newcommand{\mkInd}[3]{\mathcal{M}\{{#1}.{#2}\}({#3})}
\newcommand{\pot}[3]{\phi_{#2}({#3} : {#1})}
\newcommand{\Pot}[3]{\Phi ( #3 : \langle #1 ; #2 \rangle )}

\newcommand{\proj}[2]{\pi_{#1}}
\newcommand{\shift}{\lhd}
\newcommand{\annotshare}{\mathbin{\curlyveedownarrow}}
\newcommand{\share}[2]{\prescript{#1}{}{\curlyveedownarrow}^{#2}_{#2}}
\newcommand{\shareext}[3]{\prescript{#1}{}{\curlyveedownarrow}^{#2}_{#3}}

\newcommand{\costfree}{\mathsf{cf}}
\newcommand{\costpaid}{\mathsf{cp}}

\newcommand{\Qnn}{\mathbb{Q}_{\geq 0}}

\newcommand{\judgement}[2]{\textbf{#1} \hfill \fbox{#2}}
\newcommand{\singlestep}[4]{({#1}, {#2}) \mapsto ({#3}, {#4})}
\newcommand{\multistep}[4]{({#1}, {#2}) \mapsto^* ({#3}, {#4})}

\newcommand{\hasty}[6]{{#1} ; {#2} \vdash_{#3} {#4}:{#5} ; {#6}}
\newcommand{\subtype}{<:}
\newcommand{\Idx}{\mathcal{I}}
\newcommand{\defd}{\triangleq}
\newcommand{\idxl}{\llparenthesis}
\newcommand{\idxr}{\rrparenthesis}
\newcommand{\idxext}{\cup}